\documentclass[envcountsame]{llncs}

\usepackage{bbm}

\usepackage{amsmath, amssymb, amsthm}
\PassOptionsToPackage{hyphens}{url}\usepackage{hyperref}
\usepackage{lineno}
\usepackage{tikz}
\usetikzlibrary{shapes,calc,arrows,automata,decorations.pathmorphing}
\usepackage{xifthen}
\usepackage{xspace}
\usepackage{cleveref}
\usepackage{stackengine}
\usepackage[utf8]{inputenc}
\usepackage{enumitem}
\usepackage{mathtools}
\usepackage[noline,noend]{algorithm2e}
\usepackage{stmaryrd}
\usepackage{etoolbox}
\usepackage{mymatrix}
\usepackage{enumitem}
\usepackage{multirow}
\usepackage{graphicx}

\makeatletter
\let\xx@thm\@thm
\AtBeginDocument{\let\@thm\xx@thm}
\makeatother

\hypersetup{
    pdftitle={Termination of Triangular Integer Loops is Decidable},
    colorlinks=true,
    linkcolor=blue,
    citecolor=olive,
    filecolor=magenta,
    urlcolor=cyan
}

\newcommand\report[1]{#1}
\newcommand\paper[1]{#1}
\renewcommand\paper[1]{}

\report{\usepackage[location=appendix,appendixsectionname=Proofs]{moveproofs}}

\renewcommand{\epsilon}{\varepsilon}
\let\oldphi\phi
\let\oldvarphi\varphi
\renewcommand{\phi}{\oldvarphi}
\renewcommand{\varphi}{\oldphi}
\renewcommand{\emptyset}{\varnothing}

\newcommand{\charfun}[1]{\left\llbracket#1\right\rrbracket}
\newcommand{\vect}[1]{\overline{#1}}
\newcommand{\ZZ}{\mathbb{Z}}
\newcommand{\RR}{\mathbb{R}}
\newcommand{\NN}{\mathbb{N}}
\newcommand{\CC}{\mathcal{C}}
\newcommand{\QQ}{\mathbb{Q}}
\newcommand{\Lin}{\mathbb{A}\mathbbm{f}}
\newcommand{\PE}{\mathbb{PE}}
\newcommand{\PEN}{\mathbb{NPE}}
\newcommand{\OO}{\mathcal{O}}
\newcommand{\assign}{\ensuremath{\leftarrow}}

\newcommand{\relmiddle}[1]{\mathrel{}\middle#1\mathrel{}}

\newcommand{\algeq}[2]{
  \begin{equation}
    \label{#1}
    \begin{minipage}{0.9\textwidth}
      \algorithmstyle{plain}
      \begin{algorithm}[H]
        \centering
        #2
      \end{algorithm}
    \end{minipage}
  \end{equation}
}

\makeatletter
\newcommand{\algorithmstyle}[1]{\renewcommand{\algocf@style}{#1}}
\makeatother
\DontPrintSemicolon

\renewcommand{\arraystretch}{1.4}

\makeatletter
\renewcommand*\env@matrix[1][\arraystretch]{%
  \edef\arraystretch{#1}%
  \hskip -\arraycolsep
  \let\@ifnextchar\new@ifnextchar
  \array{*\c@MaxMatrixCols c}}
\makeatother

\spnewtheorem{ex}[theorem]{Example}{\bfseries}{\slshape}

\newcounter{auxctr}
\newcounter{eq-term-closed-form}

\DeclareMathOperator{\coeffs}{coeffs}
\DeclareMathOperator{\lia}{max\_coeff\_pos}
\DeclareMathOperator{\sign}{sign}
\DeclareMathOperator{\base}{unmark}

\Crefname{definition}{Def.}{Def.}
\Crefname{example}{Ex.}{Ex.}
\Crefname{appendix}{Appendix}{Appendix}
\Crefname{ex}{Ex.}{Ex.}
\Crefname{theorem}{Thm.}{Thm.}
\Crefname{lemma}{Lemma}{Lemmas}
\Crefname{section}{Sect.}{Sect.}
\Crefname{algorithm}{Procedure}{Procs}
\Crefname{algocf}{Alg.}{Algorithms}
\Crefname{corollary}{Cor.}{Cor.}

\title{Termination of Triangular Integer Loops is Decidable\thanks{funded by DFG grant
    389792660 as part of \href{https://perspicuous-computing.science}{TRR~248} and by DFG
    grant GI 274/6}}
\author{Florian Frohn\inst{1}\paper{\orcidID{0000-0003-0902-1994}} \and J\"urgen Giesl\inst{2}\paper{\orcidID{0000-0003-0283-8520}}}
\institute{Max Planck Institute for Informatics, Saarbr\"ucken, Germany \and
  LuFG Informatik 2, RWTH Aachen University, Germany}
\begin{document}

\setlength{\abovedisplayskip}{3pt}
\setlength{\belowdisplayskip}{3pt}
\setlength{\abovedisplayshortskip}{1pt}
\setlength{\belowdisplayshortskip}{3pt}

\maketitle

\begin{abstract}
  We consider the
  problem whether termination of affine integer
  loops  is decidable.
  Since Tiwari conjectured decidability in 2004 \cite{DBLP:conf/cav/Tiwari04}, only special cases have been solved
\cite{Bozga14,DBLP:conf/cav/Braverman06,DBLP:conf/soda/OuakninePW15}.
   We complement this work by
    proving decidability for the case that the update matrix
     is triangular.
\end{abstract}

\section{Introduction}
\label{sec:Introduction}

We consider affine integer loops of the form
\algeq{loop}{\lWhile{$\phi$}{$\vect{x} \leftarrow A\,\vect{x} + \vect{a}$.}}
Here, $A \in \ZZ^{d \times d}$ for some dimension $d \geq 1$,
$\vect{x}$ is a column vector of pairwise different variables $x_1,\ldots,x_d$, $\vect{a} \in \ZZ^d$, and $\phi$ is a
conjunction of inequalities of the form $\alpha > 0$ where $\alpha \in \Lin[\vect{x}]$ is
an affine expression
with rational
coefficients\footnote{Note that multiplying with the least common multiple of
  all denominators yields an equivalent constraint with integer coefficients,
  i.e., allowing rational instead of integer coefficients does not extend the
  considered class of loops.}  over $\vect{x}$ (i.e., $\Lin[\vect{x}] =
\{\vect{c}^T\, \vect{x} + c \mid \vect{c} \in \QQ^d, c \in \QQ\}$).
So $\phi$ has the form $B\,\vect{x} + \vect{b} > \vect{0}$ where $\vect{0}$
is the vector containing $k$ zeros, $B \in \QQ^{k \times d}$, and $\vect{b} \in \QQ^k$ for
some $k \in \NN$.
\Cref{def:term} formalizes the intuitive notion of termination for such loops.
\begin{definition}[Termination]
  \label{def:term}
  Let $f:\ZZ^d \to \ZZ^d$ with $f(\vect{x}) = A\,\vect{x} + \vect{a}$. If
  \[
    \exists \vect{c} \in \ZZ^{d}.\ \forall n \in \NN.\ \phi[\vect{x} / f^n(\vect{c})],
  \]
  then \eqref{loop} is \emph{non-terminating} and $\vect{c}$ is a \emph{witness} for
  non-termination.
  Otherwise, \eqref{loop} \emph{terminates}.
\end{definition}

Here, $f^n$ denotes the $n$-fold application of $f$, i.e., we have $f^0(\vect{c}) = \vect{c}$ and
$f^{n+1}(\vect{c}) = f(f^n(\vect{c}))$.
We call $f$ the \emph{update} of \eqref{loop}.
Moreover, for any entity
$s$, $s[x / t]$ denotes the entity that results from $s$ by replacing all occurrences of $x$ by $t$.
  Similarly, if
  $\vect{x} = \begin{sbmatrix}x_1\\[-.15cm]\vdots\\x_m\end{sbmatrix}$ and
      $\vect{t} = \begin{sbmatrix}t_1\\[-.15cm]\vdots\\t_m\end{sbmatrix}$, then
$s[\vect{x} / \vect{t}]$ denotes the entity resulting from $s$ by replacing all
      occurrences of $x_i$ by $t_i$ for each $1 \leq i  \leq m$.

\begin{ex}\label{leading ex}
  \label{ex}
  Consider the  loop

    \algorithmstyle{plain}
    \begin{algorithm}[H]
      \centering
      \lWhile{$y + z > 0$}{
        $\begin{bmatrix}
          w\\
          x\\
          y\\
          z
        \end{bmatrix}
        \assign
        \begin{bmatrix}
          2\\
          x + 1\\
          - w - 2 \cdot y\\
          x
        \end{bmatrix}$
      }
    \end{algorithm}
    \noindent
      where the update of all variables is executed simultaneously.
    This program belongs to our class of affine loops, because it can be written equivalently as follows.
    \algorithmstyle{plain}
    \begin{algorithm}[H]
      \centering
      \lWhile{$y + z > 0$}{
        $\begin{bmatrix}
          w\\
          x\\
          y\\
          z
        \end{bmatrix}
        \assign
        \begin{bmatrix*}[r]
          0&0&0&0\\
          0&1&0&0\\
          -1&\phantom{-}0&-2&\phantom{-}0\\
          0&1&0&0
        \end{bmatrix*}
        \begin{bmatrix}
          w\\
          x\\
          y\\
          z
        \end{bmatrix}
        +
        \begin{bmatrix*}[r]
          2\\
          1\\
          0\\
          0
        \end{bmatrix*}$
      }
    \end{algorithm}
\end{ex}

While termination of affine loops is known to be decidable if the variables range over the real
\cite{DBLP:conf/cav/Tiwari04} or the rational numbers \cite{DBLP:conf/cav/Braverman06}, the
integer case is a well-known open problem
\cite{DBLP:journals/toplas/Ben-AmramGM12,Bozga14,DBLP:conf/cav/Braverman06,DBLP:conf/soda/OuakninePW15,DBLP:conf/cav/Tiwari04}.\footnote{The 
  proofs for real or rational numbers do not carry over to the integers since
  \cite{DBLP:conf/cav/Tiwari04} uses Brouwer's Fixed Point Theorem which is not applicable
  if the variables range over $\ZZ$ and \cite{DBLP:conf/cav/Braverman06} relies on the
  density of $\QQ$ in $\RR$.}
However, certain special cases have been solved:
Braverman \cite{DBLP:conf/cav/Braverman06} showed that termination of \emph{linear} loops
is decidable
(i.e., loops of the form \eqref{loop} where $\vect{a}$ is $\vect{0}$ and $\phi$ is
of the form $B\,\vect{x} > \vect{0}$).
Bozga et al.\ \cite{Bozga14} showed decidability for the case that the update matrix $A$
in \eqref{loop} has the
\emph{finite monoid property}, i.e., if there is an
$n > 0$ such that $A^n$ is diagonalizable and all eigenvalues of $A^n$ are in
$\{0,1\}.$
Ouaknine et al.\ \cite{DBLP:conf/soda/OuakninePW15} proved
decidability for the
case $d \leq 4$ and for the case that $A$ is diagonalizable.

Ben-Amram et al.\ \cite{DBLP:journals/toplas/Ben-AmramGM12} showed
undecidability of termination for certain extensions of affine integer
loops, e.g., for loops where the body is of the
form $\mathbf{if}\ x > 0\ \mathbf{then}\ \vect{x} \assign A\,\vect{x}\ \mathbf{else}\ \vect{x} \assign
A'\,\vect{x}$ where $A,A' \in \ZZ^{d \times d}$ and $x \in \vect{x}$.

In this paper, we present another substantial step towards the solution of the open problem
whether termination of affine integer loops is decidable. We show that termination is
decidable
for \emph{triangular} loops \eqref{loop} where $A$ is a
triangular matrix (i.e., all entries of $A$ below or above the main diagonal are
zero).
Clearly, the order of the variables is irrelevant, i.e.,
   our results also cover the case that $A$ can be transformed into a triangular matrix by
   reordering $A$, $\vect{x}$, and $\vect{a}$ accordingly.\footnote{Similarly, one could of course
  also use other
  termination-preserving pre-processings and try to transform a
  given program into a triangular loop.}
 So essentially, triangularity means that the program variables $x_1,\ldots,x_d$ can be
 ordered such that in each loop iteration, the
 new value of $x_i$ only depends on the
previous values of $x_1,\ldots,x_{i-1},x_i$. Hence, this excludes programs with
``cyclic dependencies'' of variables (e.g., where the new values of $x$ and $y$ both depend on
the old values of both $x$ and $y$). While triangular loops are a very
restricted subclass of general integer programs, integer programs often contain such loops. Hence, tools for termination analysis of such
programs (e.g.,
\cite{DBLP:conf/tacas/BrockschmidtCIK16,KroeningTOPLAS,UltimatePLDI18,SilvaUrbanCAV15,AProVE-JAR17,LarrazFMCAD13,HipTNT})
could benefit from integrating our decision procedure and
applying it whenever a sub-program is an affine triangular loop.

Note
that triangularity and diagonalizability of matrices do not imply each other.
As we consider loops with arbitrary dimension, this means that the class of loops considered in this paper is not covered by \cite{Bozga14,DBLP:conf/soda/OuakninePW15}.
Since we consider affine instead of linear loops, it is also orthogonal to
\cite{DBLP:conf/cav/Braverman06}.

To see the difference between our and previous results, note that
a triangular matrix $A$ where
$c_1,\ldots,c_k$ are the \emph{distinct}  entries on the diagonal
is diagonalizable iff $(A - c_1 I) \ldots (A- c_k I)$ is the zero
 matrix.\footnote{The reason is that
   in this case, $(x - c_1)
 \ldots (x- c_k)$ is the minimal polynomial of $A$ and diagonalizability is equivalent to
 the fact that the minimal polynomial is a product of  distinct linear factors.}
 Here, $I$ is the identity matrix. So an easy example for a triangular loop where the
 update matrix is not diagonalizable is the following well-known program (see, e.g., \cite{DBLP:journals/toplas/Ben-AmramGM12}):
 \algorithmstyle{plain}
 \begin{algorithm}[H]
   \centering
   \lWhile{$x > 0$}{$x \assign x+y;\;  y \assign y-1$}
 \end{algorithm}
 It terminates as $y$ eventually becomes negative and then $x$ decreases in each iteration. In matrix notation, the loop body is
 $\begin{bmatrix}x\\
   y
 \end{bmatrix}
 \leftarrow
 \begin{bmatrix*}[r]
   1&1\\
   0&1
 \end{bmatrix*}
 \;
 \begin{bmatrix}x\\
   y
 \end{bmatrix}
 +
 \begin{bmatrix*}[r]0\\
   -1
 \end{bmatrix*}$,
i.e., the update matrix is triangular.
Thus, this program is in our class of programs where we show that termination is decidable.
 However,  the only entry on the diagonal of the  update matrix $A$
is $c = 1$ and
$
A - c \,I = \begin{bmatrix*}[r]
   0&1\\
   0&0
 \end{bmatrix*}
$
is not the zero matrix. So $A$ (and in fact each $A^n$ where $n \in \NN$)
is not diagonalizable. Hence, extensions of this example to a dimension
greater than $4$ where the loop is still triangular are not covered by any of the previous
results.\footnote{For instance, consider
  \lWhile{$x > 0$}{$x \assign x+y + z_1 + z_2 + z_3; \; y \assign y-1$}
  \vspace{-1.2em}}

Our proof that termination is decidable for triangular loops proceeds in three steps.
We first prove that termination of triangular loops is decidable iff
termination of \emph{non-negative triangular} loops (\emph{nnt-loops}) is decidable,
cf.\ \Cref{sec:nnt}.
A loop is non-negative if the diagonal of $A$ does not contain negative entries.
Second, we show how to compute \emph{closed forms} for nnt-loops, i.e., vectors $\vect{q}$
of $d$ expressions over the variables $\vect{x}$ and $n$ such that
$\vect{q}[n/c] = f^c(\vect{x})$ for all $c\geq 0$, see \Cref{sec:closed}.
Here,
triangularity of the matrix $A$ allows us to
treat the variables step by step. So for any $1 \leq i \leq d$, we already know the closed forms for
$x_1,\ldots,x_{i-1}$ when computing the closed form for $x_i$.
The idea of computing closed forms for the repeated updates of loops was inspired by our
previous work on
inferring lower bounds on the
runtime of integer programs \cite{ijcar2016}. But in contrast to \cite{ijcar2016}, here
the computation of the closed form always succeeds due to the restricted shape of the programs.
Finally, we explain how to decide termination of nnt-loops by reasoning about their closed
forms in \Cref{sec:deciding}.
While our technique does not yield witnesses for non-termination, we show that
it yields witnesses for \emph{eventual} non-termination, i.e., vectors $\vect{c}$ such
that $f^n(\vect{c})$ witnesses non-termination for some $n \in \NN$.
\report{All missing proofs can be found in
\Cref{app:proofs}.}\paper{Detailed proofs for all lemmas and theorems can be found in
  \cite{Report}.}

\section{From Triangular to Non-Negative Triangular Loops}
\label{sec:nnt}

To transform triangular loops into nnt-loops, we define how to \emph{chain} loops.
Intuitively, chaining yields a new loop where a single iteration is equivalent to two iterations of the original loop.
Then we show that chaining a triangular loop always yields an nnt-loop
and that chaining is equivalent w.r.t.\ termination.
\begin{definition}[Chaining]
  \label{def:chaining}
  \emph{Chaining} the loop \eqref{loop} yields:
  \algeq{chained}{
    \lWhile{$\phi \land \phi[\vect{x} / A\,\vect{x} + \vect{a}]$}{
      $\vect{x} \leftarrow A^2\,\vect{x} + A\,\vect{a} + \vect{a}$
    }
  }
\end{definition}
\begin{ex}
  \label{ex:chained}
  Chaining \Cref{ex} yields

  \algorithmstyle{plain}
  \begin{algorithm}[H]
    \While{$y + z > 0 \land - w - 2 \cdot y + x > 0$}{
      $\begin{bmatrix}
        w\\
        x\\
        y\\
        z
      \end{bmatrix}
      \assign
      \begin{bmatrix*}[r]
        0&0&0&0\\
        0&1&0&0\\
        -1&\phantom{-}0&-2&\phantom{-}0\\
        0&1&0&0
      \end{bmatrix*}^2
      \begin{bmatrix}
        w\\
        x\\
        y\\
        z
      \end{bmatrix}
      +
      \begin{bmatrix*}[r]
        0&0&0&0\\
        0&1&0&0\\
        -1&\phantom{-}0&-2&\phantom{-}0\\
        0&1&0&0
      \end{bmatrix*}
      \begin{bmatrix*}[r]
        2\\
        1\\
        0\\
        0
      \end{bmatrix*}
      +
      \begin{bmatrix*}[r]
        2\\
        1\\
        0\\
        0
      \end{bmatrix*}$
    }
  \end{algorithm}
  \noindent
  which simplifies to the following nnt-loop: \vspace*{-.2cm}

  \algorithmstyle{plain}
  \begin{algorithm}[H]
    \lWhile{$y + z > 0 \land - w - 2 \cdot y + x > 0$}{
      $\begin{bmatrix}
        w\\
        x\\
        y\\
        z
      \end{bmatrix}
      \assign
      \begin{bmatrix*}[r]
        0&0&0&0\\
        0&1&0&0\\
        2&0&4&0\\
        0&1&0&0
      \end{bmatrix*}
      \begin{bmatrix}
        w\\
        x\\
        y\\
        z
      \end{bmatrix}
      +
      \begin{bmatrix*}[r]
        2\\
        2\\
        -2\\
        1
      \end{bmatrix*}$
    }
  \end{algorithm}
\end{ex}

\report{The following lemma}\paper{\noindent{}\Cref{lem:quadratic}} is needed to
prove that \eqref{chained} is an nnt-loop if \eqref{loop} is triangular\report{ (see
\Cref{lem:quadratic_proof} for the straightforward proof of \Cref{lem:quadratic})}.
\begin{lemma}[Squares of Triangular Matrices]
  \label{lem:quadratic}
  For every triangular matrix $A$, $A^2$ is a triangular matrix whose diagonal entries are non-negative.
\end{lemma}
\report{\makeproof{lem:quadratic}{
  Let $A$ be a lower triangular matrix of dimension $d$ (the proof for upper triangular matrices is analogous).
  We have $A^2_{i;j} = \sum_{k = 1}^d A_{i;k} \cdot A_{k;j}$ (where $A_{i;j}$ is the $i^{th}$ entry in $A$'s $j^{th}$ column).

  If $i < j$, then $\sum_{k = 1}^d A_{i;k} \cdot A_{k;j} = 0$ as, for each addend, either $A_{i;k} = 0$ (if $i < k$) or $A_{k;j} = 0$ (if $k \leq i < j$), which proves that $A^2$ is triangular.

  If $i = j$, then $\sum_{k = 1}^d A_{i;k} \cdot A_{k;i} \geq 0$ as for each addend, either $A_{i;k} = 0$ (if $i < k$) or $A_{k;i} = 0$ (if $i > k$) or $A_{i;k} \cdot A_{k;i} \geq 0$ (if $i = k$), which proves that all diagonal entries of $A^2$ are non-negative.
}}
\begin{corollary}[Chaining Loops]
  \label{lem:chained-update-non-neg}
  If \eqref{loop} is triangular, then \eqref{chained} is an nnt-loop.
\end{corollary}
\begin{proof}
  Immediate consequence of \Cref{def:chaining} and \Cref{lem:quadratic}.
\end{proof}
\begin{lemma}[Equivalence of Chaining]
  \label{lem:chaining-complete}
  \eqref{loop} terminates $\iff$ \eqref{chained} terminates.
\end{lemma}
\begin{proof}
  By \Cref{def:term}, \eqref{loop} does not terminate iff
  \[
    \begin{array}{lll}
      &\exists \vect{c} \in \ZZ^{d}.\ \forall n \in \NN.\ \phi[\vect{x} / f^n(\vect{c})] & \\
      \iff&\exists \vect{c} \in \ZZ^{d}.\ \forall n \in \NN.\ \phi[\vect{x} / f^{2 \cdot
            n}(\vect{c})] \land \phi[\vect{x} / f^{2 \cdot n + 1}(\vect{c})]\\
      \iff&\exists \vect{c} \in \ZZ^{d}.\ \forall n \in \NN.\ \phi[\vect{x} / f^{2 \cdot n}(\vect{c})] \land \phi[\vect{x} / A\,f^{2 \cdot n}(\vect{c}) + \vect{a}] & (\text{by Def.\ of } f),
    \end{array}
  \]
  i.e., iff \eqref{chained} does not terminate as
  $f^2(\vect{x}) = A^2\,\vect{x} + A\,\vect{a} + \vect{a}$
  is the update of \eqref{chained}.
\end{proof}
 
\begin{theorem}[Reducing Termination  to nnt-Loops]
  \label{thm:ptill}
  Termination of triangular loops is decidable iff termination of nnt-loops is decidable.
\end{theorem}
\begin{proof}
  Immediate consequence of \Cref{lem:chained-update-non-neg} and \Cref{lem:chaining-complete}.
\end{proof}

\noindent
Thus, from now on we restrict our attention to nnt-loops.

\section{Computing Closed Forms}
\label{sec:closed}

The next step towards our decidability proof is to show that $f^n(\vect{x})$ is
equivalent
to a vector of \emph{poly-exponential expressions} for each nnt-loop, i.e., the closed
form of each nnt-loop can be represented by such expressions. Here, \emph{equivalence}
means that two expressions evaluate to the same result for all variable assignments.

Poly-exponential expressions are sums of arithmetic terms where it is always
clear which addend determines the asymptotic growth of the whole
expression when
increasing a designated variable $n$.  This is crucial for our decidability proof in
\Cref{sec:deciding}. 
Let $\NN_{\geq 1} = \{b \in \NN \mid b
\geq 1\}$ (and $\QQ_{>0}$, $\NN_{>1}$, etc.\ are defined analogously).
Moreover, $\Lin[\vect{x}]$ is again the set of all affine expressions over $\vect{x}$. \report{\pagebreak}

\begin{definition}[Poly-Exponential Expressions]
  \label{def:poly-exp}
  Let $\CC$ be the set of all finite conjunctions over the literals
    $n = c, n \neq c$ where $n$ is a designated variable and $c \in \NN$.
  Moreover for each formula $\psi$ over $n$,
    let $\charfun{\psi}$ be the characteristic function of $\psi$, i.e.,
  $\charfun{\psi}(c) = 1$ if $\psi[n/c]$ is valid and $\charfun{\psi}(c) = 0$,
  otherwise.
  The set of all \emph{poly-ex\-po\-nen\-tial expressions} over $\vect{x}$ is
  \[
    \PE[\vect{x}] = \left\{\sum_{j=1}^\ell \charfun{\psi_j} \cdot \alpha_j \cdot n^{a_j} \cdot b_j^n
    \relmiddle{|} \ell, a_j \in \NN, \; \psi_j \in \CC, \;
    \alpha_j \in \Lin[\vect{x}], \; b_j \in \NN_{\geq 1}\right\}.
  \]
\end{definition}

As $n$ ranges over $\NN$, we use $\charfun{n > c}$ as syntactic sugar for $\charfun{\bigwedge_{i=0}^{c} n \neq i}$.
So an example for a poly-exponential expression is
\[
\charfun{n > 2} \cdot (2 \cdot x + 3 \cdot y -1) \cdot n^3 \cdot 3^n \; + \; \charfun{n =
  2} \cdot (x-y).
\]
Moreover, note that
if $\psi$ contains a \emph{positive} literal (i.e., a literal of the form ``$n = c$'' for some
number $c \in \NN$), then
$\charfun{\psi}$ is equivalent to either $0$ or $\charfun{n=c}$.

The crux of the proof that poly-exponential expressions can
represent closed forms is to show that certain sums over products of exponential
and poly-ex\-po\-nential expressions can be represented by poly-exponential
expressions, cf.\ \Cref{lem:pe}. To construct these
expressions, we use a variant of \cite[Lemma 3.5]{BagnaraZZ03TR}. As
usual,
$\QQ[\vect{x}]$ is the set of all polynomials over
$\vect{x}$ with rational coefficients.

\begin{lemma}[Expressing Polynomials by Differences \cite{BagnaraZZ03TR}]
  \label{lem:pol}
  If $q \in \QQ[n]$ and $c \in \QQ$, then there is an $r \in \QQ[n]$
  such that
  $q = r - c \cdot r[n/n-1]$ for all $n \in \NN$.
\end{lemma}
\report{\makeproof{lem:pol}{
  We use induction on the degree $d$ of $q$. In the induction base, let $d =
  0$, i.e., $q = c_0 \in \QQ$.  If $c = 1$, then we fix $r = c_0 \cdot n$ and
  we get
  \[
    r - c \cdot r[n/n-1] = r - r[n/n-1] = c_0 \cdot n - c_0 \cdot (n-1) = c_0 = q.
  \]
  If $c \neq 1$, then we fix $r = \frac{c_0}{1-c}$ and we get
  \[
  r - c \cdot r[n/n-1] = \tfrac{c_0}{1-c} - c \cdot
  \tfrac{c_0}{1-c} = \tfrac{c_0}{1-c} - \tfrac{c \cdot c_0}{1 - c} = \tfrac{c_0 \cdot (1-c)}{1-c} = c_0 = q.
  \]
  For the induction step, let $d > 0$, i.e., $q = \sum_{i=0}^d c_i \cdot
  n^i$. Let $s = \frac{c_d \cdot n^{d+1}}{d+1}$ if $c = 1$ and $s = \frac{c_d \cdot n^d}{1-c}$,
  otherwise. Moreover, let $t = q - s + c \cdot s[n/n-1]$. If
  $c = 1$, then we have
  \[
    \begin{array}{lll}
      t &=& q - \frac{c_d \cdot n^{d+1}}{d+1} + \frac{c_d \cdot (n-1)^{d+1}}{d+1}\\
           &=& q - \frac{c_d \cdot n^{d+1} - c_d \cdot (n-1)^{d+1}}{d+1}\\
           &=& q - \frac{c_d \cdot n^{d+1} - c_d \cdot \sum_{k=0}^{d+1}\binom{d+1}{k} \cdot n^{d+1-k} \cdot (-1)^k}{d+1}\\
           &=& q - \frac{- c_d \cdot \sum_{k=1}^{d+1}\binom{d+1}{k} \cdot n^{d+1-k} \cdot (-1)^k}{d+1}\\
           &=& q + \frac{c_d \cdot \sum_{k=1}^{d+1}\binom{d+1}{k} \cdot n^{d+1-k} \cdot (-1)^k}{d+1}\\
           &=& q + \frac{c_d \cdot \left(\sum_{k=2}^{d+1}\binom{d+1}{k} \cdot n^{d+1-k} \cdot (-1)^k\right) - c_d \cdot (d+1) \cdot n^d}{d+1}\\
           &=& q + \frac{c_d \cdot \left(\sum_{k=2}^{d+1}\binom{d+1}{k} \cdot n^{d+1-k} \cdot (-1)^k\right)}{d+1} - c_d \cdot n^d\\
           &=& \sum_{i=0}^d c_i \cdot n^i + \frac{c_d \cdot \left(\sum_{k=2}^{d+1}\binom{d+1}{k} \cdot n^{d+1-k} \cdot (-1)^k\right)}{d+1} - c_d \cdot n^d\\
           &=& \sum_{i=0}^{d-1} c_i \cdot n^i + \frac{c_d \cdot \left(\sum_{k=2}^{d+1}\binom{d+1}{k} \cdot n^{d+1-k} \cdot (-1)^k\right)}{d+1}\\
    \end{array}
  \]
  which is a polynomial whose degree is at most $d-1$.  If $c \neq 1$, then we
  have
  \[
    \begin{array}{lll}
      t    &=& q - \frac{c_d \cdot n^d}{1-c} + c \cdot \frac{c_d \cdot (n-1)^d}{1-c}\\
           &=& q - \frac{c_d \cdot n^d}{1-c} + \frac{c \cdot c_d \cdot (n-1)^d}{1-c}\\
      &=& q - \frac{c_d \cdot n^d- c \cdot c_d \cdot (n-1)^d}{1-c}
      \end{array}\]\pagebreak
        \[
        \begin{array}{lll}
          &=& q - \frac{c_d \cdot n^d- c \cdot c_d \cdot \sum_{k=0}^{d}\binom{d}{k} \cdot n^{d-k} \cdot (-1)^k}{1-c}\\
           &=& q - \frac{c_d \cdot n^d- c \cdot c_d \cdot n^d - c \cdot c_d \cdot \sum_{k=1}^{d}\binom{d}{k} \cdot n^{d-k} \cdot (-1)^k}{1-c}\\
           &=& q - \frac{(1-c) \cdot c_d \cdot n^d - c \cdot c_d \cdot \sum_{k=1}^{d}\binom{d}{k} \cdot n^{d-k} \cdot (-1)^k}{1-c}\\
           &=& q - c_d \cdot n^d + \frac{c \cdot c_d \cdot \sum_{k=1}^{d}\binom{d}{k} \cdot n^{d-k} \cdot (-1)^k}{1-c}\\
           &=& \sum_{i=0}^d c_i \cdot n^i - c_d \cdot n^d + \frac{c \cdot c_d \cdot \sum_{k=1}^{d}\binom{d}{k} \cdot n^{d-k} \cdot (-1)^k}{1-c}\\
           &=& \sum_{i=0}^{d-1} c_i \cdot n^i + \frac{c \cdot c_d \cdot \sum_{k=1}^{d}\binom{d}{k} \cdot n^{d-k} \cdot (-1)^k}{1-c}\\
    \end{array}
  \]
  which is again a polynomial whose degree is at most $d-1$. By the induction
  hypothesis, there exists some $r' \in \QQ[n]$ such that $t = r' -
  c \cdot r'[n/n-1]$ for all $n \in \NN$. Let $r = s +
  r'$. Then we get:
  \[
    \begin{array}{ll@{\quad}l}
      &r - c \cdot r[n/n-1]\\
      =& s + r' - c \cdot s[n/n-1] - c \cdot r'[n/n-1] & (\text{by definition of } r)\\
      =& s + t - c \cdot s[n/n-1] & (\text{as } t = r' - c
  \cdot r'[n/n-1])\\
      =& t - t + q & (\text{by definition of } t)\\
      =& q
    \end{array}
  \]
}}

So \Cref{lem:pol} expresses a polynomial $q$ via the difference of
another polynomial $r$ at the positions $n$ and $n-1$, where the
additional factor $c$ can be chosen freely.
\report{A detailed proof of \Cref{lem:pol} can be found in \Cref{lem:pol_proof}. It}\paper{The proof of \Cref{lem:pol}}
is by induction on the degree of $q$ and its structure resembles the
structure of the following algorithm to compute $r$. Using the Binomial Theorem,
one can verify that  $q - s + c \cdot s[n/n-1]$ has a
smaller degree than $q$, which is crucial for the 
proof of
\Cref{lem:pol} and termination of \Cref{alg:compute-r}.

\medskip
\algorithmstyle{boxruled}
\begin{algorithm}[H]
  \KwIn{$q = \sum_{i=0}^d c_i \cdot n^i \in \QQ[n], \;\; c \in \QQ$}
  \KwResult{$r \in \QQ[n]$ such that $q = r - c \cdot r[n/n-1]$}
  \eIf{$d = 0$}{
    \leIf{$c = 1$}{\Return $c_0 \cdot n$}{\Return $\frac{c_0}{1-c}$}
  }{
    \leIf{$c = 1$}{$s \assign \frac{c_d \cdot n^{d+1}}{d+1}$}{$s \assign \frac{c_d \cdot n^d}{1-c}$}
    \Return $s + \text{compute\_}r(\, q - s + c \cdot s[n/n-1],\; c\,)$\;
  }
  \caption{compute\_$r$}
  \label{alg:compute-r}
\end{algorithm}
\medskip

\begin{ex}
  \label{ex:PolynomialDifference}
As an example, consider $q = 1$ (i.e., $c_0 = 1$) and $c = 4$. Then we search for an $r$ such
that $q = r - c \cdot r[n/n-1]$, i.e., $1 = r - 4 \cdot r[n/n-1]$. According to
\Cref{alg:compute-r}, the solution is $r = \frac{c_0}{1-c} = -\frac{1}{3}$.
  \end{ex}

\begin{lemma}[Closure of $\PE$ under Sums of Products and Exponentials]\label{lem:pe}If $m
  \in \NN$ and $p \in \PE[\vect{x}]$,
  then one can compute a
   $q \in \PE[\vect{x}]$ which is
  equivalent \report{\pagebreak} to $\sum_{i=1}^{n} m^{n - i} \cdot p[n/i-1]$.
\end{lemma}
\begin{proof}
  Let $p = \sum_{j=1}^\ell \charfun{\psi_j} \cdot \alpha_j \cdot n^{a_j} \cdot b_j^n$.
  We have:
  \begin{equation}
    \label{eq:unfold-p}
    \textstyle
    \sum\limits_{i=1}^{n} m^{n - i} \cdot p[n/i-1]
    = \sum\limits_{j=1}^\ell \sum\limits_{i=1}^{n} \charfun{\psi_j}(i-1) \cdot m^{n - i} \cdot \alpha_j \cdot (i-1)^{a_j} \cdot b_j^{i-1}
  \end{equation}
  As $\PE[\vect{x}]$ is closed under addition, it suffices to show that we can compute an equivalent poly-exponential expression
  for any expression of the form
   \begin{equation}
    \label{eq:addend}
      \textstyle
    \sum_{i=1}^{n} \; \charfun{\psi}(i-1) \cdot m^{n - i} \cdot \alpha \cdot (i-1)^{a} \cdot b^{i-1}.
 \end{equation}

    We first regard the case $m=0$. Here, the expression \eqref{eq:addend} can be simplified to
    \begin{equation}
      \label{eq:m=0-1}
      \charfun{n \neq 0 } \cdot \charfun{\psi[n/n-1]} \cdot \alpha \cdot (n-1)^{a} \cdot b^{n-1}.
    \end{equation}
    Clearly, there is a $\psi' \in \CC$ such that $\charfun{\psi'}$ is equivalent to $\charfun{n \neq 0 } \cdot \charfun{\psi[n/n-1]}$.
Moreover, $\alpha \cdot b^{n-1} = \tfrac{\alpha}{b} \cdot b^n$ where $\tfrac{\alpha}{b}
\in \Lin[\vect{x}]$. Hence, due to the Binomial Theorem
\begin{equation}
  \label{eq:m=0-2}
  \textstyle
  \charfun{n \neq 0} \cdot
  \charfun{\psi[n/n-1]} \cdot \alpha \cdot (n-1)^{a} \cdot b^{n-1}\\
  = \sum_{i=0}^a
  \charfun{\psi'} \cdot
  \tfrac{\alpha}{b}\cdot \binom{a}{i} \cdot (-1)^i
  \cdot n^{a-i}
  \cdot b^n
\end{equation}
  which is a poly-exponential expression as $\tfrac{\alpha}{b}\cdot\binom{a}{i} \cdot (-1)^i \in \Lin[\vect{x}]$.

 From now on, let $m \geq 1$.
 If $\psi$ contains a positive literal $n = c$,
  then we get
  \begin{equation}
    \label{eq:positive-literal}
    \left.
    \begin{array}{ll@{\quad}l}
     &\sum_{i=1}^{n} \charfun{\psi}(i-1) \cdot m^{n - i} \cdot \alpha \cdot (i-1)^{a}
    \cdot b^{i-1}\\
 = &\sum_{i=1}^{n} \charfun{n > i-1} \cdot \charfun{\psi}(i-1) \cdot m^{n - i} \cdot \alpha \cdot (i-1)^{a}
    \cdot b^{i-1} & (\dagger)\\
 = &\charfun{n > c} \cdot \charfun{\psi}(c) \cdot m^{n - c - 1} \cdot \alpha \cdot c^{a}
 \cdot b^{c} & (\dagger\dagger)\\
 = &\left\{ \begin{array}{l@{\quad}l}
   0,&\text{if $\charfun{\psi}(c) = 0$}\\
   \charfun{n > c} \cdot \tfrac{1}{m^{c+1}}\cdot \alpha \cdot c^a \cdot b^c  \cdot m^{n}
   ,&\text{if $\charfun{\psi}(c) = 1$}
 \end{array}
 \right.\\[.4cm]
  \in &\PE[\vect{x}] \quad \text{(since $\tfrac{1}{m^{c+1}}\cdot \alpha \cdot c^a \cdot b^c  \in
    \Lin[\vect{x}]$).}
    \end{array}
    \right\}
  \end{equation}
  The step marked with $(\dagger)$ holds as we have $\charfun{n > i-1}=1$ for
  all $i \in \{1,\ldots,n\}$ and the step marked with $(\dagger\dagger)$ holds
  since $i \neq c+1$ implies $\charfun{\psi}(i-1) = 0$.
 If $\psi$ does not contain a positive literal, then
 let $c$ be the maximal constant that occurs in $\psi$ or $-1$ if $\psi$ is empty.
We get:
\begin{equation}
  \label{eq:negative-1}
  \left.
  \begin{array}{ll@{\quad}l}
    &\sum_{i=1}^{n} \charfun{\psi}(i-1) \cdot m^{n - i} \cdot \alpha \cdot (i-1)^{a} \cdot
    b^{i-1}\\
 = &\sum_{i=1}^{n} \charfun{n > i-1} \cdot \charfun{\psi}(i-1) \cdot m^{n - i} \cdot \alpha \cdot (i-1)^{a}
    \cdot b^{i-1} & (\dagger)\\
 = &\sum_{i=1}^{c+1} \charfun{n > i-1} \cdot \charfun{\psi}(i-1) \cdot m^{n - i} \cdot \alpha \cdot (i-1)^{a}
 \cdot b^{i-1}\\
 &+ \sum_{i=c+2}^{n}   m^{n - i} \cdot \alpha \cdot (i-1)^{a}
 \cdot b^{i-1}
  \end{array}
  \right\}
\end{equation}
Again, the step marked with $(\dagger)$ holds since we have $\charfun{n > i-1}=1$
for all $i \in \{1,\ldots,n\}$. The last step holds as $i \geq c+2$ implies
$\charfun{\psi}(i-1) = 1$.  Similar to the case where $\psi$ contains a positive
literal, we can compute a poly-exponential expression which is equivalent to the
first addend. We have
\begin{equation}
  \label{eq:negative-2}
  \begin{array}{ll}
    &\sum_{i=1}^{c+1} \charfun{n > i-1} \cdot \charfun{\psi}(i-1) \cdot m^{n - i} \cdot
    \alpha \cdot (i-1)^{a} \cdot b^{i-1} \\
    =\ &\sum\limits_{\substack{1 \leq i \leq c+1 \\ \charfun{\psi}(i-1) = 1}} \charfun{n > i-1}
    \cdot \tfrac{1}{m^i} \cdot \alpha  \cdot (i-1)^{a} \cdot b^{i-1} \cdot m^n\\[.4cm]
  \end{array}
\end{equation}
which is a poly-exponential expression as $\tfrac{1}{m^{i}}\cdot \alpha \cdot (i-1)^a \cdot b^{i-1}  \in \Lin[\vect{x}]$.
For the second addend, we have:
\begin{equation}
  \label{eq:negative-3}
  \left.
    \begin{array}{@{\hspace*{-.6cm}}lll}
      &\sum_{i=c+2}^{n} m^{n - i} \cdot \alpha \cdot (i-1)^{a} \cdot b^{i-1}\\
      =&\frac{\alpha}{b} \cdot m^n \cdot \sum_{i=c+2}^{n} (i-1)^{a} \cdot \left(\frac{b}{m}\right)^i\\
      =&\frac{\alpha}{b} \cdot m^n \cdot \sum_{i=c+2}^{n} (r[n/i] - \frac{m}{b} \cdot
      r[n/i-1]) \cdot \left(\frac{b}{m}\right)^{i} \hfill \text{(\Cref{lem:pol} with
        $c = \frac{m}{b}$)\hspace*{-.1cm}}\\
      =&\frac{\alpha}{b} \cdot m^n \cdot \left(\sum_{i=c+2}^{n} r[n/i] \cdot \left(\frac{b}{m}\right)^{i} - \sum_{i=c+2}^{n} \frac{m}{b} \cdot r[n/i-1] \cdot \left(\frac{b}{m}\right)^{i}\right)\\
      =&\frac{\alpha}{b} \cdot m^n \cdot \left(\sum_{i=c+2}^{n} r[n/i] \cdot \left(\frac{b}{m}\right)^{i} - \sum_{i=c+1}^{n-1} r[n/i] \cdot \left(\frac{b}{m}\right)^{i}\right)\\
      =&\frac{\alpha}{b} \cdot m^n \cdot \charfun{n > c+1} \cdot (r \cdot
      \left(\frac{b}{m}\right)^{n} - r[n/c+1] \cdot \left(\frac{b}{m}\right)^{c+1})\\
      =&\charfun{n > c+1} \cdot \frac{\alpha}{b} \cdot r \cdot b^{n} - \charfun{n > c+1} \cdot r[n/c+1] \cdot \left(\frac{b}{m}\right)^{c+1} \cdot \frac{\alpha}{b} \cdot m^n
    \end{array}
    \right\}\hspace*{-.5cm}
  \end{equation}
  \Cref{lem:pol} ensures $r \in \QQ[n]$, i.e., we have $r =
  \sum_{i=0}^{d_r} m_i \cdot n^i$ for some $d_r \in \NN$ and $m_i \in \QQ$. Thus,
  $r[n/c+1] \cdot \left(\frac{b}{m}\right)^{c+1} \cdot \frac{\alpha}{b} \in
  \Lin[\vect{x}]$ which implies $\charfun{n > c+1} \cdot r[n/c+1] \cdot
  \left(\frac{b}{m}\right)^{c+1} \cdot \frac{\alpha}{b} \cdot m^n \in \PE[\vect{x}]$. It
  remains to show that the addend $\charfun{n > c+1} \cdot \frac{\alpha}{b} \cdot r \cdot
  b^{n}$ is equivalent to a poly-exponential expression. As $\frac{\alpha}{b} \cdot m_i \in \Lin[\vect{x}]$, we have
  \begin{equation}
    \label{eq:negative-4}
    \textstyle
    \charfun{n > c+1} \cdot \frac{\alpha}{b} \cdot r \cdot b^{n}
    = \sum_{i=0}^{d_r} \charfun{n > c+1} \cdot \frac{\alpha}{b} \cdot m_i \cdot n^i \cdot b^{n}
    \in  \PE[\vect{x}].
  \end{equation}
\end{proof}

\noindent
The proof of \Cref{lem:pe} gives rise to a corresponding algorithm.

\medskip
\algorithmstyle{boxruled}
\begin{algorithm}[H]
  \KwIn{$m \in \NN, \;\; p \in \PE[\vect{x}]$}
  \KwResult{$q \in \PE[\vect{x}]$ which is equivalent to $\sum_{i=1}^n m^{n-i} \cdot p[n/i-1]$}
  rearrange $\sum_{i=1}^n m^{n-i} \cdot p[n/i-1]$ to $\sum_{j=1}^\ell p_j$ as in \eqref{eq:unfold-p}\;
  \ForEach{$p_j \in \{p_1,\ldots,p_\ell\}$}{
    \lIf{$m = 0$}{compute $q_j$ as in \eqref{eq:m=0-1} and \eqref{eq:m=0-2}}
    \lElseIf{$p_j = \charfun{\ldots \land n = c \land \ldots} \cdot \ldots$}{compute $q_j$ as in \eqref{eq:positive-literal}}
    \Else{
      \begin{itemize}[label=$\bullet$]
      \item split $p_j$ into two sums $p_{j,1}$ and $p_{j,2}$ as in \eqref{eq:negative-1}
      \item compute $q_{j,1}$ from $p_{j,1}$ as in \eqref{eq:negative-2}
      \item compute $q_{j,2}$ from $p_{j,2}$ as in \eqref{eq:negative-3} and \eqref{eq:negative-4} using \Cref{alg:compute-r}
      \item $q_j \assign q_{j,1} + q_{j,2}$
      \end{itemize}\vspace{-\topsep} 
    }
  }
  \Return $\sum_{j=1}^\ell q_j$\;
  \caption{symbolic\_sum}
  \label{alg:closed-forms-for-symbolic-sums}
\end{algorithm}
\medskip

\begin{ex}
  \label{ex:sum}
 We compute an equivalent poly-exponential expression for
  \begin{equation}
    \label{original}
    \textstyle
    \sum_{i=1}^{n} 4^{n-i} \cdot (\charfun{n=0} \cdot 2 \cdot w \; + \;
    \charfun{n \neq 0} \cdot 4
    \; - \; 2)\;[n/i-1]
  \end{equation}
  where $w$ is a variable. (It will later on be needed to compute a closed form for \Cref{ex:chained}, see \Cref{ex:closed}.) According to
\Cref{alg:closed-forms-for-symbolic-sums} and \eqref{eq:unfold-p}, we get
  \[
    \begin{array}{lll}
      & \sum_{i=1}^{n} 4^{n-i} \cdot (\charfun{n=0} \cdot 2 \cdot w \; + \;
      \charfun{n \neq 0} \cdot 4 \; - \; 2)\;[n/i-1]\\
      =& \sum_{i=1}^{n} 4^{n-i} \cdot (\charfun{i-1=0} \cdot 2 \cdot w \; + \;
      \charfun{i-1 \neq 0} \cdot 4 \; - \; 2)\\
      =& p_1 + p_2 + p_3
    \end{array}\]
    with $p_1 = \sum_{i=1}^{n}
    \charfun{i-1=0} \cdot 4^{n-i} \cdot 2 \cdot w$, $p_2 = \sum_{i=1}^{n}
    \charfun{i-1 \neq 0} \cdot 4^{n-i} \cdot 4$, and $p_3 =
    \sum_{i=1}^{n} 4^{n-i} \cdot (- 2)$.
    We search for $q_1, q_2,
    q_3 \in \PE[w]$ that are equivalent to $p_1, p_2, p_3$, i.e., $q_1 + q_2 + q_3$
    is equivalent to \eqref{original}. We only show how to
    compute $q_2$\report{. See \Cref{ex:sum-full} for the computation of $q_1 = \charfun{n \neq 0}
    \cdot \tfrac{1}{2} \cdot w \cdot 4^{n}$ and $q_3 = \tfrac{2}{3} - \tfrac{2}{3} \cdot
    4^n$.}\paper{ (and omit the computation of  $q_1 = \charfun{n \neq 0}
    \cdot \tfrac{1}{2} \cdot w \cdot 4^{n}$ and $q_3 = \tfrac{2}{3} - \tfrac{2}{3} \cdot
    4^n$).}
    Analogously to \eqref{eq:negative-1}, we get:
    \[
    \begin{array}{lllr}
      &&    \sum_{i=1}^{n}
\charfun{i-1 \neq 0} \cdot 4^{n-i} \cdot 4 & \\
&=&    \sum_{i=1}^{n} \charfun{n > i-1} \cdot
\charfun{i-1 \neq 0} \cdot 4^{n-i} \cdot 4\\
&=&    \sum_{i=1}^{1} \charfun{n > i-1} \cdot
  \charfun{i-1 \neq 0} \cdot 4^{n-1} \cdot 4
  \quad + \quad
    \sum_{i=2}^{n}  4^{n-i} \cdot 4
    \end{array}
  \]
  The next step is to rearrange the first sum as in \eqref{eq:negative-2}. In our example, it directly simplifies to
  $0$ and hence we obtain
  \[
    \textstyle
    \sum_{i=1}^{1} \charfun{n > i-1} \cdot \charfun{i-1 \neq 0} \cdot 4^{n-1} \cdot 4 + \sum_{i=2}^{n}  4^{n-i} \cdot 4 = \sum_{i=2}^{n}  4^{n-i} \cdot 4.
  \]
  Finally, by applying the steps from \eqref{eq:negative-3} we get:
  \[
    \begin{array}{lllrr}
      && \sum_{i=2}^{n}  4^{n-i} \cdot 4 & \\
&=& 4 \cdot 4^n \cdot
   \sum_{i=2}^{n}  \left(\tfrac{1}{4}\right)^i &\\
&=& 4 \cdot 4^n \cdot
   \sum_{i=2}^{n}
 \left(-\tfrac{1}{3} - 4 \cdot
\left(-\tfrac{1}{3}\right)\right) \cdot \left(\tfrac{1}{4}\right)^i &\qquad (\dagger)\\
&=&  4 \cdot 4^n \cdot \left(
\sum_{i=2}^{n} \left(-\tfrac{1}{3}\right)  \cdot \left(\tfrac{1}{4}\right)^i \; - \;
\sum_{i=2}^{n} 4 \cdot \left(-\tfrac{1}{3}\right)  \cdot \left(\tfrac{1}{4}\right)^i
\right)\\
&=&  4 \cdot 4^n \cdot \left(
\sum_{i=2}^{n} \left(-\tfrac{1}{3}\right)  \cdot \left(\tfrac{1}{4}\right)^i \; - \;
\sum_{i=1}^{n-1}  \left(-\tfrac{1}{3}\right)  \cdot \left(\tfrac{1}{4}\right)^i
\right)\\
   &=& 4 \cdot 4^n \cdot   \charfun{n > 1} \cdot
\left(\left(-\tfrac{1}{3}\right)  \cdot \left(\tfrac{1}{4}\right)^n \; - \;
 \left(-\tfrac{1}{3}\right) \cdot \tfrac{1}{4}
 \right)\\
 &=& \charfun{n > 1} \cdot \left(-\tfrac{4}{3}\right) \; + \;  \charfun{n > 1} \cdot
    \tfrac{1}{3} \cdot 4^n\\
      &=&q_2
    \end{array}\]
The step marked with $(\dagger)$ holds by \Cref{lem:pol} with $q = 1$ and $c = 4$. Thus, we have $r =
  -\tfrac{1}{3}$, cf.\ \Cref{ex:PolynomialDifference}.
\end{ex}

Recall that our goal is to compute closed forms for loops. As a first step,
instead of the $n$-fold
update function $h(n,\vect{x}) = f^n(\vect{x})$ of \eqref{loop} where $f$ is the
update of \eqref{loop}, we
consider a recursive update function for a single variable $x \in \vect{x}$:
\[
  \textstyle
  g(0,\vect{x}) = x \quad \text{and} \quad g(n,\vect{x}) = m \cdot g(n-1, \vect{x}) + p[n/n-1] \quad \text{for all $n > 0$}
\]
Here,
$m \in \NN$ and $p \in \PE[\vect{x}]$.
Using \Cref{lem:pe}, it is easy to show that $g$
can be represented by a poly-exponential expression.
\begin{lemma}[Closed Form for Single Variables]
  \label{thm:closed-under-linear-polynomials}
  If $x \in \vect{x}$, $m \in \NN$, and $p \in \PE[\vect{x}]$, then one can compute a
  $\,q \in \PE[\vect{x}]$ which
  satisfies
  \[
    \textstyle
    q\,[n/0] = x \quad \text{and} \quad q = (m \cdot q + p)\;[n/n-1] \quad \text{for all } n > 0.
  \]
\end{lemma}
\begin{proof}
   It suffices to find a $q \in \PE[\vect{x}]$
  that satisfies
  \begin{equation}\label{eq:closed-non-zero}
    \textstyle
    q = m^n \cdot x + \sum_{i=1}^{n} m^{n-i} \cdot p[n/i-1].
  \end{equation}
   To see why \eqref{eq:closed-non-zero} is sufficient, note that
\eqref{eq:closed-non-zero} implies
\[
  \textstyle
  q[n/0] \quad = \quad m^0 \cdot x + \sum\nolimits_{i=1}^{0} m^{0-i} \cdot p[n/i-1]
  \quad =\quad x
\]
    and for $n > 0$, \eqref{eq:closed-non-zero} implies \pagebreak
    \[
    \begin{array}{llll}
      q &=& m^{n}  \cdot x + \sum_{i=1}^{n} m^{n-i} \cdot p[n/i-1]\\
             &=& m^{n} \cdot x + \left(\sum_{i=1}^{n-1} m^{n-i} \cdot p[n/i-1]\right) + p[n/n-1]\\
             &=& m \cdot \left(m^{n-1} \cdot x + \sum_{i=1}^{n-1} m^{n-i-1} \cdot p[n/i-1]\right) + p[n/n-1]\\
             &=& m \cdot q[n/n-1] + p[n/n-1]\\
             &=& (m \cdot q + p)[n/n-1].
    \end{array}
  \]
  By \Cref{lem:pe}, we can compute a $q' \in \PE[\vect{x}]$ such that
  \[
    \textstyle
    m^n \cdot x + \sum_{i=1}^{n} m^{n-i} \cdot p[n/i-1] \quad =
    \quad  m^n \cdot x + q'.
  \]
  Moreover,
\vspace*{-.2cm}
  \begin{align}
    \label{eq:closed-non-zero-1}
  &\text{if $m = 0$, then $m^n \cdot x = \charfun{n = 0} \cdot x \in \PE[\vect{x}]$ and}\\
    \label{eq:closed-non-zero-2}
    &\text{if $m > 0$, then $m^n \cdot x \in \PE[\vect{x}]$.}
  \end{align}
  So both addends are equivalent to poly-exponential expressions.
\end{proof}

\begin{ex}
  \label{ex:single-closed}
  We show how to compute the closed forms for the variables $w$ and $x$ from \Cref{ex:chained}.
  We first consider the assignment $w \assign 2$,
i.e., we want to compute a $q_w \in \PE[w,x,y,z]$ with $q_w [n/0] = w$ and $q_w = (m_w
\cdot q_w + p_w)\,[n/n-1]$ for $n > 0$, where $m_w = 0$ and $p_w = 2$. According to
\eqref{eq:closed-non-zero}  and \eqref{eq:closed-non-zero-1}, $q_w$ is
  \[\textstyle
  m_w^n \cdot w + \sum_{i=1}^{n} m_w^{n-i} \cdot p_w[n/i-1]  =
  0^n \cdot w + \sum_{i=1}^{n} 0^{n-i} \cdot 2  =
    \charfun{n = 0} \cdot w + \charfun{n \neq 0} \cdot 2.
  \]
  For the assignment $x \assign x + 2$,
we search for a $q_x$ such that $q_x[n/0] = x$ and $q_x  =  (m_x
\cdot q_x + p_x)\,[n/n-1]$ for $n > 0$, where
$m_x = 1$ and $p_x = 2$. By \eqref{eq:closed-non-zero}, $q_x$ is
  \[\textstyle
  m_x^n \cdot x + \sum_{i=1}^{n} m_x^{n-i} \cdot p_x[n/i-1]   =
    1^n \cdot x + \sum_{i=1}^{n} 1^{n-i} \cdot 2
     =  x + 2
    \cdot n.
  \]
\end{ex}

The restriction to triangular matrices now allows us to
generalize \Cref{thm:closed-under-linear-polynomials} to vectors of variables.
The reason is that due to triangularity,
the update of each program variable $x_i$ only depends on the previous values of $x_1,\ldots,x_{i}$. So when
regarding $x_i$, we can assume that we already know
the closed forms for $x_1,\ldots,x_{i-1}$.
This allows us to find closed forms for one variable after the other by
applying \Cref{thm:closed-under-linear-polynomials} repeatedly.
In other words, it allows us to find a vector $\vect{q}$ of poly-exponential
expressions that satisfies
\[
  \textstyle
  \vect{q}\,[n/0] = \vect{x}\quad \text{and} \quad \vect{q} = A\, \vect{q}[n/n-1] + \vect{a} \quad \text{for all } n > 0.
\]
To prove this claim, we show the more general \Cref{lem:polynomial-solutions}.
For all $i_1,\ldots,i_k \in \{1, \ldots, m\}$, we define
$[z_1,\ldots,z_m]_{i_1,\ldots,i_k} = [z_{i_1},\ldots,z_{i_k}]$ (and the notation
$\vect{y}_{i_1,\ldots,i_k}$ for column vectors is defined analogously).
Moreover, for a matrix $A$, $A_{i}$ is $A$'s $i^{th}$ row and
$A_{i_1,\ldots,i_n;j_1,\ldots,j_k}$ is the matrix with rows
$(A_{i_1})_{j_1,\ldots,j_k}, \ldots, (A_{i_n})_{j_1,\ldots,j_k}$. So for $A =
\begin{bmatrix}
  a_{1,1}&a_{1,2}&a_{1,3}\\
  a_{2,1}&a_{2,2}&a_{2,3}\\
  a_{3,1}&a_{3,2}&a_{3,3}
\end{bmatrix}$,
we have $A_{1,2;1,3} = \begin{bmatrix}
  a_{1,1}&a_{1,3}\\
  a_{2,1}&a_{2,3}
\end{bmatrix}$.

\begin{lemma}[Closed Forms for Vectors of Variables]
  \label{lem:polynomial-solutions}
  If $\vect{x}$ is a vector of at least $d \geq 1$ pairwise different variables, 
  $A \in \ZZ^{d \times d}$ is triangular with $A_{i;i} \geq  0$ for all $1 \leq i \leq d$,
  and $\vect{p} \in \PE[\vect{x}]^d$,
  then one can compute  $\vect{q} \in \PE[\vect{x}]^d$ such that:
  \begin{align}
    \label{eq:closed-vector-1}
    \vect{q}\,[n/0] &= \vect{x}_{1,\ldots,d}\quad \text{and}\\
    \label{eq:closed-vector-2}
    \vect{q} &= (A\, \vect{q} + \vect{p})\;[n/n-1] \quad \text{for all } n > 0
  \end{align}
\end{lemma}
\begin{proof}
  Assume that $A$ is lower triangular \pagebreak
  (the case that $A$ is upper triangular works analogously).
  We use induction on $d$.
   For any $d \geq 1$ we have:
  \[
    \begin{array}{llllll}
      &\vect{q} &=& (A\, \vect{q} + \vect{p})\;[n/n-1]\\
      \iff& \vect{q}_j &=& (A_{j} \cdot \vect{q} + \vect{p}_j)\;[n/n-1] & \text{for all } 1 \leq j \leq d\\
      \iff& \vect{q}_j &=& (A_{j;2,\ldots,d} \cdot \vect{q}_{2,\ldots,d} + A_{j;1} \cdot \vect{q}_1 + \vect{p}_j)\;[n/n-1] & \text{for all } 1 \leq j \leq d\\
      \iff& \vect{q}_1 &=& (A_{1;2,\ldots,d} \cdot \vect{q}_{2,\ldots,d} + A_{1;1} \cdot \vect{q}_1 + \vect{p}_1)\;[n/n-1] & \land\\
      & \vect{q}_j &=& (A_{j;2,\ldots,d} \cdot \vect{q}_{2,\ldots,d} + A_{j;1} \cdot \vect{q}_1 + \vect{p}_j)\;[n/n-1] & \text{for all } 1 < j \leq d\\
      \iff& \vect{q}_1 &=& (A_{1;1} \cdot \vect{q}_1 + \vect{p}_1)\;[n/n-1] & \land\\
      & \vect{q}_j &=& (A_{j;2,\ldots,d} \cdot \vect{q}_{2,\ldots,d} + A_{j;1} \cdot \vect{q}_1 + \vect{p}_j)\;[n/n-1] & \text{for all } 1 < j \leq d
    \end{array}
  \]
  The last step holds as $A$ is lower triangular.
  By \Cref{thm:closed-under-linear-polynomials}, we can compute a $\vect{q}_1 \in \PE[\vect{x}]$
  that satisfies
  \[
    \textstyle
    \vect{q}_1[n/0] = \vect{x}_1 \quad \text{and} \quad \vect{q}_1 = (A_{1;1} \cdot \vect{q}_1 + \vect{p}_1)\;[n/n-1] \quad \text{for all } n > 0.
  \]

  In the induction base ($d = 1$), there is no $j$ with $1 < j \leq d$. In
  the induction step ($d > 1$),
it remains to show that we can compute $\vect{q}_{2,\ldots,d}$ such that
\[
  \textstyle
  \vect{q}_j[n/0] = \vect{x}_j \quad \text{and} \quad
  \vect{q}_j = (A_{j;2,\ldots,d} \cdot \vect{q}_{2,\ldots,d} + A_{j;1} \cdot \vect{q}_1
  + \vect{p}_j)\;[n/n-1]
\]
for all $n > 0$ and all $1 < j \leq d$, which is equivalent to
  \begin{align*}
  \vect{q}_{2,\ldots,d}[n/0] &= \vect{x}_{2,\ldots,d} \quad \text{and}\\[-1.3em]
    \vect{q}_{2,\ldots,d} &= (A_{2,\ldots,d;2,\ldots,d} \cdot \vect{q}_{2,\ldots,d}
    + \begin{bmatrix}A_{2;1}\\\vdots\\A_{d;1}\end{bmatrix}
     \cdot \vect{q}_1 + \vect{p}_{2,\ldots,d})\;[n/n-1]
  \end{align*}
  for all $n>0$.
  As $A_{j;1} \cdot \vect{q}_1 + \vect{p}_j \in \PE[\vect{x}]$ for each $2 \leq j \leq d$, the claim follows from the induction hypothesis.
 \end{proof}

Together, \Cref{thm:closed-under-linear-polynomials,lem:polynomial-solutions}
and their proofs give rise to the following algorithm to compute a solution for
\eqref{eq:closed-vector-1} and \eqref{eq:closed-vector-2}.
It computes a closed form $\vect{q}_1$ for $\vect{x}_1$ as in the proof of \Cref{thm:closed-under-linear-polynomials}, constructs the argument
$\vect{p}$ for the recursive call based on $A$, $\vect{q}_1$, and the current
value of $\vect{p}$ as in the proof of \Cref{lem:polynomial-solutions}, and then determines the closed form for $\vect{x}_{2, \ldots, d}$ recursively.

\medskip
\algorithmstyle{boxruled}
\begin{algorithm}[H]
  \KwIn{$\vect{x}_{1,\ldots,d}, \;\; A \in \ZZ^{d \times d}$ where $A_{i;i} \geq 0$ for all $1 \leq i
    \leq d$, $\;\; \vect{p} \in \PE[\vect{x}]^d$}
  \KwResult{$\vect{q} \in \PE[\vect{x}]^d$ which satisfies \eqref{eq:closed-vector-1} \& \eqref{eq:closed-vector-2} for the given $\vect{x}, A$, and $\vect{p}$}
  $q \assign \text{symbolic\_sum}(A_{1;1},\vect{p}_1)$ (cf.\ \Cref{alg:closed-forms-for-symbolic-sums})\;
  \leIf{$A_{1;1} = 0$}{$\vect{q}_1 \assign \charfun{n=0} \cdot \vect{x}_1 + q$}{$\vect{q}_1 \assign A_{1;1}^n \cdot \vect{x}_1 + q$ (cf.\ (\ref{eq:closed-non-zero}--\ref{eq:closed-non-zero-2}))}
  \If{$d>1$}{
    \vspace{-0.4cm}
    $\vect{q}_{2,\ldots,d} \assign
    \text{closed\_form}(\vect{x}_{2,\ldots,d}, A_{2,\ldots,d;2,\ldots,d},
    \begin{bmatrix}A_{2;1}\\\vdots\\A_{d;1}\end{bmatrix} \cdot
    \vect{q}_1 + \vect{p}_{2,\ldots,d})$
    \vspace{-0.4cm}
  }
  \Return $\vect{q}$\;
  \caption{closed\_form}
  \label{alg:closed-form-vector}
\end{algorithm}
\medskip

\noindent
We can now prove the main theorem of this section.

\begin{theorem}[Closed Forms for nnt-Loops]
  \label{lem:closed}
  One can compute a closed form for every nnt-loop. \pagebreak In other words, if
$f:\ZZ^d \to \ZZ^d$ is the update function of an nnt-loop with the variables $\vect{x}$,
then one can compute a $\vect{q} \in \PE[\vect{x}]^d$ such that $\vect{q}[n/c] = f^c(\vect{x})$ for all
$c \in \NN$.
\end{theorem}
\begin{proof}
Consider an nnt-loop of the form \eqref{loop}.
By \Cref{lem:polynomial-solutions}, we can compute a
$\vect{q} \subseteq \PE[\vect{x}]^d$ that satisfies
\[
  \textstyle
  \vect{q}[n/0] = \vect{x} \quad \text{and} \quad \vect{q} = (A\, \vect{q} + \vect{a})\;[n/n-1] \quad \text{for all } n > 0.
\]
  We prove $f^c(\vect{x}) = \vect{q}[n/c]$ by induction on $c \in \NN$.
  If $c=0$, we get
  \[
    f^c(\vect{x}) = f^0(\vect{x}) = \vect{x} = \vect{q}[n/0] = \vect{q}[n/c].
  \]
  \[
    \begin{array}{l@{\hspace*{1cm}}llll}
   \mbox{If $c>0$, we get:}&   f^c(\vect{x}) &=& A\, f^{c-1}(\vect{x}) + \vect{a} & \text{by definition of } f\\
                       &&=& A\, \vect{q}[n/c-1] + \vect{a} & \text{by the induction hypothesis}\\
                       &&=& (A\, \vect{q} + \vect{a})\;[n/c-1] & \text{as } \vect{a} \in
      \ZZ^d \text{ does not contain } n\\
                       &&=& \vect{q}[n/c] &
    \end{array}    \vspace*{-.3cm}
    \]
\end{proof}

\noindent
So invoking \Cref{alg:closed-form-vector} on $\vect{x}, A$, and $\vect{a}$ yields the closed form of an nnt-loop \eqref{loop}.

\begin{ex}
  \label{ex:closed}
  We show how to compute the closed form for \Cref{ex:chained}.
  For
  \[
    y \assign 2 \cdot w + 4 \cdot y - 2,
  \]
  we obtain
  \[
    \begin{array}{lll@{\;\,}l}
      q_y&=& \left( 4 \cdot q_y + 2 \cdot q_w - 2 \right) \; [n / n-1]\\
&=&      4^n \cdot y + \sum_{i=1}^{n} 4^{n-i} \cdot (2 \cdot q_w - 2)\; [n/i-1]& (\text{by
        \eqref{eq:closed-non-zero}})\\
        &=& y \cdot 4^n  + \sum_{i=1}^{n} 4^{n-i} \cdot (\charfun{n=0} \cdot 2 \cdot w +
      \charfun{n \neq 0} \cdot 4 - 2)\;[n/i-1] & (\text{see  \Cref{ex:single-closed}})\\
      &=&  q_0   + q_1 + q_2 + q_3 & (\text{see  \Cref{ex:sum}})
    \end{array}
  \]
  where $q_0 = y \cdot 4^n$. For $z \assign x + 1$, we get
  \[
    \begin{array}{lll@{\hspace*{1cm}}l}
      q_z &=& (q_x + 1)\; [n / n-1]\\
      &=& 0^n \cdot z +  \sum_{i=1}^{n} 0^{n-i} \cdot (q_x + 1)\;[n/i-1]& (\text{by
        \eqref{eq:closed-non-zero}})\\
 &=& \charfun{n = 0} \cdot z + \charfun{n \neq 0} \cdot (q_x[n/n-1] + 1)\\
      &=& \charfun{n = 0} \cdot z + \charfun{n \neq 0} \cdot \left((x + 2 \cdot n)\;[n/n-1] + 1\right)
& (\text{see  \Cref{ex:single-closed}})\\
            &=& \charfun{n = 0} \cdot z + \charfun{n \neq 0} \cdot (x - 1) + \charfun{n \neq 0} \cdot 2 \cdot n.
    \end{array}
  \]
  So the closed form of \Cref{ex:chained} for the values of the variables
  after $n$ iterations is:
  \[
    \footnotesize \begin{bmatrix}
      q_w\\
      q_x\\
      q_y\\
      q_z
    \end{bmatrix}
    =
    \begin{bmatrix}
      \charfun{n=0} \cdot w + \charfun{n \neq 0} \cdot 2\\
      x + 2 \cdot n\\
      q_0  + q_1 + q_2 + q_3\\
      \charfun{n = 0} \cdot z + \charfun{n \neq 0} \cdot (x - 1) + \charfun{n \neq 0} \cdot 2 \cdot n
    \end{bmatrix}
  \]
    \end{ex}

\section{Deciding Non-Termination of nnt-Loops}
\label{sec:deciding}

Our
proof uses the notion of \emph{eventual non-termination} \cite{DBLP:conf/cav/Braverman06,DBLP:conf/soda/OuakninePW15}.
Here, the idea is to disregard the condition of the loop during a finite prefix of the
program run.\report{\pagebreak}
\begin{definition}[Eventual Non-Termination]
  \label{def:non-term}
  A vector $\vect{c} \in \ZZ^d$ \emph{witnesses eventual non-termination} of
  \eqref{loop} if  \paper{\pagebreak}
  \[
    \exists n_0 \in \NN.\ \forall n \in \NN_{>n_0}.\ \phi[\vect{x} / f^{n}(\vect{c})].
  \]
  If there is such a witness, then \eqref{loop} is \emph{eventually non-ter\-mi\-nat\-ing}.
\end{definition}
Clearly, \eqref{loop} is non-terminating iff \eqref{loop} is
eventually non-terminating \cite{DBLP:conf/soda/OuakninePW15}\report{ (see
\Cref{app:non-term} for a proof)}.
Now \Cref{lem:closed} gives rise to an alternative characterization of eventual non-ter\-mi\-nat\-ion in terms of the closed form $\vect{q}$ instead of $f^{n}(\vect{c})$.
\begin{corollary}[Expressing Non-Termination with $\PE$]
  \label{thm:term-closed-form}
  If $\vect{q}$ is the closed form of \eqref{loop}, then $\vect{c} \in \ZZ^d$ witnesses eventual non-termination iff
  \begin{equation}
    \label{eq:term-closed-form-1}
    \exists n_0 \in \NN.\ \forall n \in \NN_{>n_0}.\ \phi[\vect{x} / \vect{q}][\vect{x} / \vect{c}].
  \end{equation}
 \end{corollary}
\report{\makeproof{thm:term-closed-form}{
  We have:
  \[
    \begin{array}{lll}
      &\vect{c} \text{ witnesses eventual non-termination}\\
      \iff&\exists n_0 \in \NN.\ \forall n \in \NN_{>n_0}.\ \phi[\vect{x} / f^{n}(\vect{c})]\\
      \iff&\exists n_0 \in \NN.\ \forall n \in
      \NN_{>n_0}.\ \phi[\vect{x} / \vect{q}[\vect{x} / \vect{c}]]& (\text{as } \vect{q}
      \text{ is the closed form of } \eqref{loop}\\
&&\text{and thus, $\vect{q} = f^n(\vect{x})$, cf.\ \Cref{lem:closed}})\\
      \iff&\eqref{eq:term-closed-form-1}
    \end{array}
  \]
}}
\begin{proof}
  Immediate, as $\vect{q}$ is equivalent to $f^n(\vect{x})$.
  \report{See
  \Cref{thm:term-closed-form_proof} for details.}
\end{proof}

So to prove that termination of nnt-loops is decidable, we will use \Cref{thm:term-closed-form} to show that the existence of a witness for eventual non-termination is decidable.
To do so, we first eliminate the factors $\charfun{\psi}$ from the closed form $\vect{q}$.
Assume that $\vect{q}$ has at least one factor $\charfun{\psi}$ where $\psi$ is non-empty (otherwise,
all factors $\charfun{\psi}$ are equivalent to $1$) and
let
$c$ be the maximal constant that occurs in such a factor.
Then all addends $\charfun{\psi} \cdot \alpha \cdot n^{a} \cdot b^n$ where $\psi$ contains a positive literal become $0$ and all other addends become $\alpha \cdot n^{a} \cdot b^n$ if $n > c$.
Thus, as we can assume $n_0 > c$ in \eqref{eq:term-closed-form-1} without loss of generality, all factors $\charfun{\psi}$ can be eliminated when checking eventual non-termination.

\begin{corollary}[Removing $\charfun{\psi}$ from $\PE$s]
  \label{cor:normalized}
  Let $\vect{q}$ be the closed form of an nnt-loop \eqref{loop}.
  Let
  $\vect{q}_{norm}$ result from $\vect{q}$ by removing all addends
  $\charfun{\psi} \cdot \alpha \cdot n^{a} \cdot b^n$ where $\psi$ contains a positive
  literal and by replacing all addends $\charfun{\psi} \cdot \alpha \cdot n^{a} \cdot b^n$ where $\psi$ does not contain a positive literal by $\alpha \cdot n^{a} \cdot b^n$.
  Then $\vect{c} \in \ZZ^d$ is a witness for eventual non-termination iff
  \begin{equation}
    \label{eq:term-closed-form-witness}
    \exists n_0 \in \NN.\ \forall n \in \NN_{>n_0}.\ \phi[\vect{x} /
      \vect{q}_{norm}][\vect{x} / \vect{c}].
  \end{equation}
  \setcounter{eq-term-closed-form}{\value{equation}}
 \end{corollary}
\report{\makeproof{cor:normalized}{
  By \Cref{thm:term-closed-form}, $\vect{c}$ is a witness for eventual non-termination iff
  \[
    \exists n_0 \in \NN.\ \forall n \in \NN_{>n_0}.\ \phi[\vect{x} / \vect{q}][\vect{x} / \vect{c}].
  \]
  Let $c$ be the maximal constant that occurs in a sub-expression $\charfun{\psi}$ in $\vect{q}$.
  Then
  \[
    \begin{array}{llll}
      &\exists n_0 \in \NN.\ \forall n \in \NN_{>n_0}.\ \phi[\vect{x} / \vect{q}][\vect{x} / \vect{c}]\\
      \iff&\exists n_0 \in \NN_{>c}.\ \forall n \in \NN_{>n_0}.\ \phi[\vect{x} / \vect{q}][\vect{x} / \vect{c}]\\
      \iff&\exists n_0 \in \NN_{>c}.\ \forall n \in \NN_{>n_0}.\ \phi[\vect{x} / \vect{q}_{norm}][\vect{x} / \vect{c}]\\
      \iff&\exists n_0 \in \NN.\ \forall n \in \NN_{>n_0}.\ \phi[\vect{x} / \vect{q}_{norm}][\vect{x} / \vect{c}].
    \end{array}
  \]
}
\begin{proof}
  See above for the proof idea and \Cref{cor:normalized_proof} for a detailed proof.
\end{proof}}

By removing the factors $\charfun{\psi}$ from the closed form $\vect{q}$ of an nnt-loop, we obtain \emph{normalized} poly-exponential expressions.

\begin{definition}[Normalized $\PE$s]
  We call $p \in \PE[\vect{x}]$ \emph{normalized} if it is in
  \[
    \textstyle
    \PEN[\vect{x}] = \left\{\sum_{j=1}^\ell \alpha_j \cdot n^{a_j} \cdot b_j^n \relmiddle{|}
    \ell, a_j \in \NN, \; \alpha_j \in \Lin[\vect{x}], \;  b_j \in \NN_{\geq 1}\right\}.
  \]
  W.l.o.g., we always assume $(b_i,a_i) \neq (b_j,a_j)$ for all
  $i,j \in \{1,\ldots,\ell\}$ with $i \neq j$.
  We define $\PEN = \PEN[\emptyset]$, i.e., we have $p \in \PEN$ if $\alpha_j \in \QQ$ for all $1 \leq j \leq \ell$.
\end{definition}

\begin{ex}
  \label{ex:normalized}
  We continue \Cref{ex:closed}.
  By omitting the factors $\charfun{\psi}$,
  \[
    \begin{array}{lll@{\quad}lll}
      q_w &=& \charfun{n=0} \cdot w + \charfun{n \neq 0} \cdot 2 &\text{becomes}& 2,\\
      q_z &=& \charfun{n = 0} \cdot z + \charfun{n \neq 0} \cdot (x - 1) + \charfun{n \neq
        0} \cdot 2 \cdot n &\text{becomes}&  x - 1 + 2 \cdot n,
    \end{array}
  \]
  and
  $q_x =  x + 2 \cdot n,
  q_0 = y \cdot 4^n$, and
  $q_3 = \tfrac{2}{3} - \frac{2}{3} \cdot 4^{n}$
  remain unchanged. Moreover,
  \[
    \begin{array}{lllllll}
      q_1 &=& \charfun{n \neq 0} \cdot \frac{1}{2} \cdot w
      \cdot 4^{n}  & \text{becomes}& \frac{1}{2} \cdot w \cdot 4^{n} &\text{ and}\\
      q_2 &=& \charfun{n > 1} \cdot \left(- \frac{4}{3}\right) +
\charfun{n > 1} \cdot \frac{1}{3}\cdot 4^{n}
&\text{becomes}&\left(- \frac{4}{3}\right) + \frac{1}{3}\cdot 4^{n}.
    \end{array}
  \]
  Thus, $q_y = q_0 + q_1 + q_2 + q_3$ becomes
  \[
    \textstyle
    y \cdot 4^n + \frac{1}{2} \cdot w \cdot 4^{n} - \frac{4}{3} + \frac{1}{3}\cdot
    4^{n} + \tfrac{2}{3}- \frac{2}{3} \cdot 4^{n}
    = 4^n \cdot \left(y - \frac{1}{3} + \frac{1}{2} \cdot w\right) - \frac{2}{3}.
  \]
  Let $\sigma = \left[w/2,\, x/x+ 2 \cdot n, \, y/4^n \cdot \left(y - \frac{1}{3} + \frac{1}{2}
      \cdot w\right) - \frac{2}{3}, \,  z/x-1 + 2 \cdot n\right]$.
   Then we get that \Cref{ex} is non-terminating iff there are $w,x,y,z \in \ZZ, n_0 \in
      \NN$ such that
      \[
        \begin{array}{l}
          (y + z)\;\sigma > 0 \land (- w - 2 \cdot y + x)\; \sigma > 0 \hfill \iff\\
          4^n \cdot \left(y - \frac{1}{3} + \frac{1}{2} \cdot w\right) - \frac{2}{3} + x
          - 1 + 2 \cdot n  > 0 \hfill \land \hfill \\
          \hspace{2em} - 2 - 2 \cdot \left(4^n \cdot \left(y - \frac{1}{3} + \frac{1}{2} \cdot w\right) -
          \frac{2}{3}\right) + x + 2 {\cdot} n  > 0 \hfill \iff\\
          p^{\phi}_1 > 0 \land p^{\phi}_2 > 0\\
        \end{array}
      \]
      holds for all $n > n_0$ where
      \[
        \begin{array}{llll}
          p^{\phi}_1 &=& 4^n \cdot \left(y - \frac{1}{3} + \frac{1}{2} \cdot w\right) + 2 \cdot n + x - \frac{5}{3} & \text{and}\\
          p^{\phi}_2 &=& 4^n \cdot \left(\frac{2}{3} - 2 \cdot y - w\right) + 2 \cdot n + x - \frac{2}{3}.
        \end{array}
      \]
\end{ex}

Recall that the loop condition
$\phi$ is a conjunction of inequalities of the form $\alpha > 0$ where $\alpha \in \Lin[\vect{x}]$.
Thus, $\phi[\vect{x} / \vect{q}_{norm}]$ is a conjunction of inequalities $p > 0$ where $p
\in \PEN[\vect{x}]$
and we need to decide if there is an instantiation of these
in\-equalities that is valid ``for large enough $n$''.
To do so, we order the coefficients $\alpha_j$ of the addends $\alpha_j \cdot n^{a_j}
\cdot b_j^n$ of normalized poly-exponential expressions according to the addend's
asymptotic growth when increasing $n$. \Cref{lem:domination} shows that
$\alpha_2 \cdot n^{a_2} \cdot b_2^n$
grows faster than  $\alpha_1 \cdot n^{a_1} \cdot b_1^n$ iff $b_2 > b_1$ or both $b_2 = b_1$
and $a_2 > a_1$.

\begin{lemma}[Asymptotic Growth]
  \label{lem:domination}
  Let $b_1,b_2 \in \NN_{\geq 1}$ and $a_1, a_2 \in \NN$.
  If $(b_2, a_2) >_{lex} (b_1, a_1)$, then $\OO(n^{a_1} \cdot b_1^n) \subsetneq \OO(n^{a_2} \cdot b_2^n)$. Here, ${>_{lex}}$ is the lexicographic order, i.e., $(b_2,a_2) >_{lex}
  (b_1,a_1)$ iff $b_2 > b_1$ or $b_2 = b_1 \land a_2 > a_1$.
\end{lemma}
\report{\makeproof{lem:domination}{
  Recall that $f,g: \NN \to \QQ$, $f(n) \in \OO(g(n))$ means
  \[
    \exists m \in \NN_{>0}, n_0 \in \NN.\ \forall n \in \NN_{> n_0}.\ |f(n)| \leq m \cdot |g(n)|.
  \]
  First consider the case $b_2 > b_1$.
  We have $b_2^n = b_1^n \cdot \left(\tfrac{b_2}{b_1}\right)^n$ where $\frac{b_2}{b_1} > 1$.
  As we clearly have $n^{a_1} \in \OO\left(\left(\tfrac{b_2}{b_1}\right)^n\right)$, we
  obtain
  $n^{a_1} \cdot b_1^n \in \OO\left(\left(\tfrac{b_2}{b_1}\right)^n \cdot b_1^n\right) =
  \OO(b_2^n) \subseteq \OO(n^{a_2} \cdot b_2^n)$, i.e., $\OO(n^{a_1} \cdot  b_1^n) \subseteq \OO(n^{a_2} \cdot b_2^n)$.

  To prove that the inclusion is strict, we show $b_2^n \notin \OO(n^{a_1} \cdot b_1^n)$.
  We have
  \[
    \begin{array}{lll}
      &b_2^n \notin \OO(n^{a_1} \cdot b_1^n)\\
      \iff &\forall m \in \NN_{>0}, n_0 \in \NN.\ \exists n \in \NN_{> n_0}.\ |b_2^n| > m \cdot |n^{a_1} \cdot b_1^n|\\
      \iff &\forall m \in \NN_{>0}, n_0 \in \NN.\ \exists n \in \NN_{> n_0}.\ b_2^n > m \cdot n^{a_1} \cdot b_1^n\\
      \iff &\forall m \in \NN_{>0}, n_0 \in \NN.\ \exists n \in \NN_{> n_0}.\ \left(\tfrac{b_2}{b_1}\right)^n > m \cdot n^{a_1} \\
      \iff &\left(\tfrac{b_2}{b_1}\right)^n \notin \OO(n^{a_1})
    \end{array}
  \]
  which holds as $\frac{b_2}{b_1} > 1$.

  Now consider the case $b_2 = b_1$ and $a_2 > a_1$.
  Then $\OO(n^{a_1} \cdot b_1^n) \subseteq \OO(n^{a_2} \cdot b_2^n)$ trivially holds.
  To prove that the inclusion is strict, we show $n^{a_2} \cdot b_2^n \notin \OO(n^{a_1} \cdot b_1^n)$.
  We have
  \[
    \begin{array}{llll}
      &n^{a_2} \cdot b_2^n \notin \OO(n^{a_1} \cdot b_1^n)\\
      \iff &\forall m \in \NN_{>0}, n_0 \in \NN.\ \exists n \in \NN_{> n_0}.\ |n^{a_2} \cdot b_2^n| > m \cdot |n^{a_1} \cdot b_1^n|\\
      \iff &\forall m \in \NN_{>0}, n_0 \in \NN.\ \exists n \in \NN_{> n_0}.\ n^{a_2} \cdot b_2^n > m \cdot n^{a_1} \cdot b_1^n\\
      \iff &\forall m \in \NN_{>0}, n_0 \in \NN.\ \exists n \in \NN_{> n_0}.\ n^{a_2} > m \cdot n^{a_1} & (\text{as } b_1 = b_2)\\
      \iff &n^{a_2} \notin \OO(n^{a_1})
    \end{array}
  \]
  which holds as $a_2 > a_1$.
}}
\begin{proof}
  By considering the cases $b_2 > b_1$ and $b_2 = b_1$ separately, the claim can
  easily be deduced from the definition of $\OO$.
  \report{See
  \Cref{lem:domination_proof} for details.}
\end{proof}

\begin{definition}[Ordering Coefficients]
  \label{def:marked}
  \emph{Marked coefficients} are of the form $\alpha^{(b,a)}$ where $\alpha \in \Lin[\vect{x}], b \in \NN_{\geq 1}$, and $a \in \NN$.
  We define $\base(\alpha^{(b,a)}) = \alpha$ and $\alpha_2^{(b_2,a_2)} \succ \alpha_1^{(b_1,a_1)}$ if $(b_2,a_2) >_{lex}
  (b_1,a_1)$.
  Let
  \[
    \textstyle
    p = \sum_{j=1}^\ell \alpha_j \cdot n^{a_j} \cdot b_j^n \in \PEN[\vect{x}],
  \]
  where $\alpha_j \neq 0$ for all $1 \leq j \leq \ell$. The marked coefficients of $p$ are
  \[
    \coeffs(p) = \begin{cases}
      \left\{0^{(1,0)}\right\}, & \text{if } \ell = 0\\
      \left\{\alpha_j^{(b_j,a_j)} \relmiddle{|} 0 \leq j \leq \ell\right\}, & \text{otherwise.}
    \end{cases}
  \]
\end{definition}
\begin{ex}
  \label{ex:marked}
  In \Cref{ex:normalized} we saw that the loop from \Cref{ex} is non-terminating iff there are
  $w,x,y,z \in \ZZ, n_0 \in \NN$ such that
  $p^{\phi}_1 > 0 \land p^{\phi}_2 > 0$
  for all $n > n_0$.
  We get:
  \begin{align*}
    \coeffs\left(p^{\phi}_1\right) &= \left\{\left(y - \tfrac{1}{3} + \tfrac{1}{2} \cdot w\right)^{(4,0)}, 2^{(1,1)}, \left(x-\tfrac{5}{3}\right)^{(1, 0)}\right\}\\
    \coeffs\left(p^{\phi}_2\right) &= \left\{\left(\tfrac{2}{3} - 2 \cdot y - w\right)^{(4,0)}, 2^{(1,1)}, \left(x-\tfrac{2}{3}\right)^{(1,0)}\right\}
  \end{align*}
\end{ex}

Now it is easy to see that the asymptotic growth of a normalized poly-exponential expression is solely determined by its $\succ$-maximal addend.
\begin{corollary}[Maximal Addend Determines Asymptotic Growth]
  \label{cor:theta}
  Let $p \in \PEN$ and let
  $\max_{\succ}(\coeffs(p)) = c^{(b,a)}$.
  Then $\OO(p) = \OO(c \cdot n^a \cdot b^n)$.
\end{corollary}
\report{\makeproof{cor:theta}{
  If $p = 0$, then $c = 0$ by \Cref{def:marked} and hence $\OO(p) = \OO(c \cdot n^a \cdot b^n) = \OO(0)$.
  Otherwise, $p$ has the form
  \[
  \begin{array}{l} c \cdot n^a \cdot b^n + \sum_{j=1}^\ell c_j \cdot n^{a_j} \cdot
    b_i^n \end{array}
  \]
for $c \neq 0$ and $\ell \geq 0$.
  We have $c_j^{(b_j,a_j)} \in \coeffs(p)$ and hence $(b,a) >_{lex} (b_j,a_j)$ for all $1 \leq j \leq \ell$.
  Thus, \Cref{lem:domination} implies $\OO(n^{a_j} \cdot b_j^n) \subsetneq \OO(n^a \cdot b^n)$ and hence we get
  \[
   \begin{array}{l}  \OO(p) = \OO\left(c \cdot n^a \cdot b^n + \sum_{j=1}^\ell c_j \cdot
     n^{a_j} \cdot b_j^n\right) = \OO(n^a \cdot b^n)= \OO(c \cdot n^a \cdot b^n). \end{array}
  \]
}}
\begin{proof}
  Clear, as $c \cdot n^a \cdot b^n$ is the asymptotically dominating addend of $p$.
  \report{See \Cref{cor:theta_proof} for a detailed proof.}
\end{proof}

Note that \Cref{cor:theta} would be incorrect for the case $c = 0$ if we replaced $\OO(p) = \OO(c \cdot n^a \cdot b^n)$ with $\OO(p) = \OO(n^a \cdot b^n)$ as $\OO(0) \neq \OO(1)$.
Building upon \Cref{cor:theta}, we now show that, for large $n$, the sign of a
normalized poly-exponential expression is solely determined by its $\succ$-maximal
coefficient.
Here, we define $\sign(c) = -1$ if $c \in \QQ_{<0} \cup \{-\infty\}$, $\sign(0) = 0$, and
$\sign(c) = 1$ if $c \in \QQ_{>0} \cup \{\infty\}$.

\begin{lemma}[Sign of $\PEN$s]
  \label{lem:asym}
  Let $p \in \PEN$.
  Then $\lim_{n \mapsto \infty} p \in \QQ$ iff
    $p \in \QQ$ and otherwise, $\lim_{n \mapsto
  \infty} p \in \{ \infty, -\infty \}$.
Moreover, we have
\[
  \textstyle
  \sign\left(\lim_{n \mapsto \infty} p\right) = \sign(\base(\max_{\succ}(\coeffs(p)))).
\]
\end{lemma}
\report{\makeproof{lem:asym}{
  If $p = 0$, \pagebreak then $\lim_{n \mapsto \infty} p = 0$ and
  $\base(\max_{\succ}(\coeffs(p))) = 0$ and hence the claim holds.
  If $p \in \QQ\setminus\{0\}$,  then $p = c \cdot n^{0} \cdot 1^n$ for some $c \in
  \QQ\setminus\{0\}$
  and hence $\lim_{n \mapsto \infty} p = c$ and
  \[
  \begin{array}{l}  \sign\left(\lim_{n \mapsto \infty} p\right) = \sign(c) =
    \sign(\base(\max_\succ(\coeffs(p)))). \end{array}
  \]

 Now consider the case $p \notin \QQ$. Note that we have $c \neq 0$ for all $c^{(b,a)} \in
 \coeffs(p)$ and thus, $\sign(\base(\max_{\succ}(\coeffs(p)))) \in \{1,-1\}$. Hence, it
 suffices to prove
  \[
    \begin{array}{l} \lim_{n \mapsto \infty} p = \sign(\base(\max_{\succ}(\coeffs(p)))) \cdot \infty. \end{array}
  \]
  Let
  \[
   \begin{array}{l}  p = \sum_{j=1}^\ell c_j \cdot n^{a_j} \cdot b_j^n \end{array}
    \]
    with $\ell \geq 1$ and $c_j \neq 0$ for all $1 \leq j \leq \ell$, since $p \neq 0$.
  We use induction on $\ell$.
  In the induction base ($\ell = 1$), we have
  $\max_{\succ}(\coeffs(p)) = c_1^{(b_1,a_1)}$ and the claim follows as we have
  \[
    \begin{array}{llllll}
      \lim_{n \mapsto \infty} c_1 \cdot n^{a_1} \cdot b_1^n &=& \phantom{-}\infty &\iff& c_1 > 0 & \text{and}\\
      \lim_{n \mapsto \infty} c_1 \cdot n^{a_1} \cdot b_1^n &=& -\infty &\iff& c_1 < 0.
    \end{array}
  \]

  In the induction step, we have $\ell > 1$.
  Let
  \[
   \begin{array}{l}   \max_{\succ}\left(\coeffs\left(\sum_{j = 2}^\ell c_j \cdot n^{a_j}
     \cdot b_j^n\right)\right) = c^{(b,a)}. \end{array}
  \]
  The induction hypothesis implies
  \begin{equation}
    \label{eq:ih}
    \begin{array}{llll}
      \lim_{n \mapsto \infty} c_1 \cdot n^{a_1} \cdot b_1^n &=&
      \left\{ \begin{array}{ll}
        \sign(c_1) \cdot \infty, &\text{if $a_1 \neq 0$ or $b_1 \neq 1$}\\
        c_1, &\text{otherwise}
        \end{array} \right.\\[0.4cm]
      \lim_{n \mapsto \infty} \sum_{j = 2}^\ell c_j \cdot n^{a_j} \cdot b_j^n &=&
  \left\{ \begin{array}{ll}
        \sign(c) \cdot \infty, &\text{if $a \neq 0$ or $b \neq 1$}\\
        c, &\text{otherwise}
  \end{array} \right.
      \end{array}
  \end{equation}
  Moreover, we have:
  \begin{equation}
    \label{eq:theta}
    \begin{array}{llllll}
     \OO\left(c_1 \cdot n^{a_1} \cdot b_1^n\right) &&
     &=& \OO(n^{a_1} \cdot b_1^n)& \text{and}\\
     \OO\left(\sum_{j = 2}^\ell c_j \cdot n^{a_j} \cdot b_j^n\right) &=& \OO(c \cdot n^{a} \cdot b^n) &=& \OO(n^{a} \cdot b^n)
    \end{array}
  \end{equation}
  by \Cref{cor:theta}, as $c_1 \neq 0$ and $c \neq 0$.
  If $c_1^{(b_1,a_1)} \succ c^{(b,a)}$, then $a_1 \neq 0$ or $b_1 \neq 1$. Moreover, then
  we have
  \begin{align}
    \OO(n^{a} \cdot b^n) &\subsetneq \OO(n^{a_1} \cdot b_1^n) && \text{(by  \Cref{lem:domination}) $\;$ and}\label{eq:domination}\\
      \max\nolimits_{\succ}(\coeffs(p)) &= c_1^{(b_1,a_1)}&&(\text{by  \Cref{def:marked}}) \label{eq:max}
  \end{align}
  and hence we obtain
  \[
    \begin{array}{lll}
      &\lim_{n \mapsto \infty} p\\
      =& \lim_{n \mapsto \infty} \left(c_1 \cdot n^{a_1} \cdot b_1^n + \sum_{j = 2}^\ell c_j \cdot n^{a_j} \cdot b_j^n\right)\\
      =& \lim_{n \mapsto \infty} c_1 \cdot n^{a_1} \cdot b_1^n & (\text{by } \eqref{eq:theta} \text{ and } \eqref{eq:domination})\\
      =& \sign(c_1) \cdot \infty & (\text{by } \eqref{eq:ih} \text{ as } a_1 \neq 0 \text{ or } b_1 \neq 1)\\
      =& \sign(\base(\max_{\succ}(\coeffs(p)))) \cdot \infty &(\text{by } \eqref{eq:max}),
    \end{array}
  \]
  as desired.
  The case $c^{(b,a)} \succ c_1^{(b_1,a_1)}$ is analogous.
}}
\begin{proof}
  If $p \notin \QQ$, then the limit of each addend of $p$ is in $\{-\infty,
  \infty\}$ by definition of $\PEN$. As the asymptotically dominating addend
  determines $\lim_{n \mapsto \infty} p$ and $\base(\max_{\succ}(\coeffs(p)))$
  determines the sign of the asymptotically dominating addend, the claim
  follows.
  \report{See \Cref{lem:asym_proof} for a detailed proof.}
\end{proof}

\Cref{thm:limits} shows the connection between the limit of a normalized poly-expo\-nen\-tial expression $p$
and the question whether  $p$ is positive for
large enough $n$. The latter corresponds to the existence of a witness for eventual
non-termination by
\Cref{cor:normalized} as $\phi[\vect{x} / \vect{q}_{norm}]$ is a conjunction of inequalities $p > 0$ where $p
\in \PEN[\vect{x}]$.

\begin{lemma}[Limits and Positivity of $\PEN$s]
  \label{thm:limits}
  Let $p \in \PEN$.
  Then
  \[
    \textstyle
    \exists  n_0 \in \NN.\ \forall n \in
    \NN_{>n_0}.\ p > 0 \iff
    \lim_{n \mapsto \infty}
    p > 0.
  \]
\end{lemma}
\report{\makeproof{thm:limits}{
  First note that $\lim_{n \mapsto \infty} p$ exists due to \Cref{lem:asym}.
  It remains to show that
  \[
    \exists n_0 \in \NN.\ \forall n \in \NN_{>n_0}.\ p > 0 \iff \lim_{n \mapsto \infty} p > 0.
  \]

  If $\lim_{n \mapsto \infty} p = \infty$, then
  \[
    \begin{array}{lll}
      &\forall m \in \QQ.\ \exists n_0 \in \NN.\ \forall n \in \NN_{> n_0}.\ p > m\\
      \implies & \exists n_0 \in \NN.\ \forall n \in \NN_{>n_0}.\ p > 0.
    \end{array}
  \]

  If $\lim_{n \mapsto \infty} p = -\infty$, then
  \[
    \begin{array}{lll}
      &\forall m \in \QQ.\ \exists n_0 \in \NN.\ \forall n \in \NN_{> n_0}.\ p < m\\
      \implies &\exists n_0 \in \NN.\ \forall n \in \NN_{> n_0}.\ p < 0\\
      \implies &\forall n_0 \in \NN.\ \exists n \in \NN_{>n_0}.\ p < 0\\
      \implies &\forall n_0 \in \NN.\ \exists n \in \NN_{>n_0}.\ p \leq 0\\
        \iff &\neg(\exists n_0 \in \NN.\ \forall n \in \NN_{>n_0}.\ p > 0)
    \end{array}
  \]

  If $\lim_{n \mapsto \infty} p \in \QQ_{> 0}$, then $p \in \QQ_{> 0}$ by \Cref{lem:asym}
  and thus
  \[
    \begin{array}{l}
         \exists n_0 \in \NN.\ \forall n \in \NN_{>n_0}.\ p > 0
    \end{array}
  \]

  If $\lim_{n \mapsto \infty} p \in \QQ_{\leq 0}$, then $p \in \QQ_{\leq 0}$ by \Cref{lem:asym}
  and thus
  \[
    \begin{array}{ll}
      & \forall n_0 \in \NN.\ \exists n \in \NN_{>n_0}.\ p \leq 0\\
      \iff &\neg(\exists n_0 \in \NN.\ \forall n \in \NN_{>n_0}.\ p > 0)
    \end{array}
  \]
}}
\begin{proof}
  By case analysis over $\lim_{n \mapsto \infty} p$.
  \report{See
  \Cref{thm:limits_proof} for details.}
\end{proof}

Now we show that \Cref{cor:normalized} allows us to decide eventual non-termination by
examining the coefficients of normalized poly-exponential expressions.
As these coefficients are in $\Lin[\vect{x}]$, the required reasoning is
 decidable. \report{\pagebreak}
\begin{lemma}[Deciding Eventual Positiveness of $\PEN$s]
  \label{lem:decidable}
  Validity of
  \begin{equation}
    \label{eq:decidable}
  \begin{array}{l}  \exists \vect{c} \in \ZZ^{d}, n_0 \in \NN.\ \forall n \in \NN_{>n_0}.\
    \bigwedge_{i=1}^k p_i[\vect{x}/\vect{c}] > 0 \end{array}
  \end{equation}
  where $p_1,\ldots,p_k \in \PEN[\vect{x}]$ is decidable.
\end{lemma}
\begin{proof}
For
any $p_i$ with $1 \leq i \leq k$
and any $\vect{c} \in \ZZ^{d}$, we have $p_i[\vect{x}/\vect{c}] \in \PEN$. Hence:
  \[
    \begin{array}{ll@{\qquad}l}
      & \exists  n_0 \in \NN.\ \forall n \in \NN_{>n_0}.\
      \bigwedge_{i=1}^k p_i[\vect{x}/\vect{c}] > 0\\
     \iff  &  \bigwedge_{i=1}^k \exists  n_0 \in \NN.\ \forall n \in \NN_{>n_0}.\
      p_i[\vect{x}/\vect{c}] > 0\\
    \iff  &  \bigwedge_{i=1}^k \lim_{n \mapsto \infty} p_i[\vect{x}/\vect{c}] > 0 & (\text{by  \Cref{thm:limits}})\\
            \iff & \bigwedge_{i=1}^k
            \base(\max_{\succ}(\coeffs(p_i[\vect{x}/\vect{c}]))) > 0 & (\text{by
              \Cref{lem:asym}})
      \end{array}
    \]
 Let $p \in \PEN[\vect{x}]$
 with $\coeffs(p) = \left\{\alpha_1^{(b_1,a_1)}\!,\ldots,\alpha^{(b_{\ell},a_{\ell})}_{\ell}\right\}$
 where $\alpha^{(b_i,a_i)}_i \succ \alpha^{(b_{j},a_{j})}_{j}$ for all $1 \leq i < j \leq
 \ell$.
 If
  $p[\vect{x}/\vect{c}] = 0$ holds, then $\coeffs(p[\vect{x}/\vect{c}]) = \{ 0^{(1,0)}
  \}$ and thus $\base(\max_{\succ}(\coeffs(p[\vect{x}/\vect{c}]))) = 0$. Otherwise,
there is an $1 \leq j \leq \ell$ with $\base(\max_{\succ}(\coeffs(p[\vect{x}/\vect{c}]))) = \alpha_j[\vect{x}/\vect{c}] \neq 0$ and we have
$\alpha_i[\vect{x}/\vect{c}] = 0$ for all $1 \leq i \leq j-1$.
 Hence, $\base(\max_{\succ}(\coeffs(p[\vect{x}/\vect{c}]))) > 0$ holds iff
 $\bigvee_{j=1}^\ell \left(\alpha_j[\vect{x}/\vect{c}] > 0 \land \bigwedge_{i=0}^{j-1}
 \alpha_i[\vect{x}/\vect{c}] = 0\right)$ holds, i.e., iff $[\vect{x}/\vect{c}]$ is a model for
       \begin{equation}
         \label{eq:lia0}
         \begin{array}{l} \lia(p) = \bigvee_{j=1}^\ell \left(\alpha_j > 0 \land \bigwedge_{i=0}^{j-1}
           \alpha_i = 0\right). \end{array}
       \end{equation}

    \noindent
    Hence by the considerations above,  \eqref{eq:decidable} is valid iff
    \begin{equation}
      \label{eq:qlia}
      \begin{array}{l}
      \exists \vect{c} \in \ZZ^{d}. \;
      \bigwedge_{i=1}^k \lia(p_i) [\vect{x}/\vect{c}]
      \end{array}
    \end{equation}
    is valid.
    By multiplying each (in-)equality in \eqref{eq:qlia}
    with the least common multiple of all denominators,
    one obtains a
    first order formula over the theory of linear integer arithmetic.
    It is well known that validity of such formulas is decidable.
  \end{proof}
\noindent
Note that \eqref{eq:qlia} is valid iff $\bigwedge_{i=1}^k \lia(p_i)$
is satisfiable. So to implement our decision procedure, one can use integer programming or
SMT solvers to check satisfiability of $\bigwedge_{i=1}^k \lia(p_i)$.
\Cref{lem:decidable} allows us to prove our main theorem.

\begin{theorem}
  \label{thm:decidable}
  Termination of triangular loops is decidable.
\end{theorem}
\begin{proof}
  \setcounter{auxctr}{\value{equation}}
  \setcounter{equation}{\value{eq-term-closed-form}}
  By \Cref{thm:ptill}, termination of triangular loops is decidable iff
  termination of nnt-loops is decidable. For an nnt-loop \eqref{loop}
  we obtain a $\vect{q}_{norm} \in \PEN[\vect{x}]^{d}$ (see \Cref{lem:closed} and \Cref{cor:normalized})
  such that \eqref{loop} is non-terminating iff
  \begin{equation}
    \label{eq:decidable2}
        \exists \vect{c} \in \ZZ^{d}, n_0 \in \NN.\ \forall n \in \NN_{>n_0}.\ \phi[\vect{x} /
      \vect{q}_{norm}][\vect{x} / \vect{c}],
  \end{equation}
  where $\phi$ is a conjunction of inequalities of the form $\alpha > 0$, $\alpha \in
  \Lin[\vect{x}]$. Hence,
\[\begin{array}{l}  \phi[\vect{x} /
      \vect{q}_{norm}][\vect{x} / \vect{c}] \; = \; \bigwedge_{i=1}^k
p_i[\vect{x}/\vect{c}] > 0
\end{array}\]
where $p_1,\ldots,p_k \in \PEN[\vect{x}]$. Thus,
by \Cref{lem:decidable}, validity of \eqref{eq:decidable2} is decidable.
\end{proof}

\noindent
The following algorithm summarizes our decision procedure.

\medskip
\algorithmstyle{boxruled}
\begin{algorithm}[H]
  \KwIn{a triangular loop \eqref{loop}}
  \KwResult{$\top$ if \eqref{loop} terminates, $\bot$ otherwise}
  \begin{itemize}[label=$\bullet$]
  \item apply \Cref{def:chaining} to \eqref{loop}, i.e.,\\
    \begin{itemize}
    \item[] $\phi \assign \phi \land \phi[\vect{x} / A\,\vect{x} + \vect{a}]$
    \item[] $A \assign A^2$
    \item[] $\vect{a} \assign A\,\vect{a} + \vect{a}$
    \end{itemize}
  \item $\vect{q} \assign \text{closed\_form}(\vect{x}, A, \vect{a})$ (cf.\ \Cref{alg:closed-form-vector})
  \item compute $\vect{q}_{norm}$ as in \Cref{cor:normalized}
  \item compute $\phi[\vect{x}/\vect{q}_{norm}] = \bigwedge_{i=1}^k p_i > 0$
  \item compute $\varphi = \bigwedge_{i=1}^k \lia(p_i)$ (cf.\ \eqref{eq:lia0})
  \item \leIf{$\varphi$ is satisfiable}{\Return $\bot$}{\Return $\top$}
  \end{itemize}\vspace{-1em}
  \caption{Deciding Termination of Triangular Loops}
  \label{alg:deciding}
\end{algorithm}
\medskip

\setcounter{equation}{\value{auxctr}}
\begin{ex}
  In \Cref{ex:marked} we showed that \Cref{ex}
  is non-ter\-mi\-nat\-ing iff
  \[
    \textstyle
    \exists w,x,y,z \in \ZZ,\ n_0 \in \NN.\ \forall n \in \NN_{>n_0}.\ p^{\phi}_1 > 0 \land p^{\phi}_2 > 0
  \]
  is valid.
   This is the case iff $\lia(p_1) \land \lia(p_2)$, i.e.\medskip,

  \noindent\medskip
  \resizebox{\textwidth}{!}{
    $
    \renewcommand{\arraystretch}{1}
    \begin{array}{c}
      y - \frac{1}{3} + \frac{1}{2} {\cdot} w > 0 \lor 2 > 0 \land y - \frac{1}{3} + \frac{1}{2} {\cdot} w = 0 \lor x-\frac{5}{3} > 0 \land 2 = 0 \land y - \frac{1}{3} + \frac{1}{2} {\cdot} w = 0\\
      \land\\
      \frac{2}{3} - 2 {\cdot} y - w > 0 \lor 2 > 0 \land \frac{2}{3} - 2 {\cdot} y - w = 0 \lor x-\frac{2}{3} > 0 \land 2 = 0 \land \frac{2}{3} - 2 {\cdot} y - w = 0
    \end{array}
    $
  }

  \noindent
  is satisfiable. This formula is equivalent to $6 \cdot y - 2 + 3 \cdot w = 0$ which does not have any
    integer solutions. Hence,
 the loop of
  \Cref{ex} terminates.
\end{ex}

\Cref{ex:characterizing-non-terminating-inputs} shows that our technique does not yield witnesses for
non-termination, but it only proves the existence of a witness for
\emph{eventual} non-termination. While such a witness can be transformed into a
witness for non-termination by applying the loop several times, it is unclear
how often the loop needs to be applied.

\begin{ex}\label{Characterizing Non-Terminating Inputs}
  \label{ex:characterizing-non-terminating-inputs}
  Consider the following non-terminating loop:
  \algorithmstyle{plain}
  \algeq{eq:nt-loop}{
    \lWhile{$x > 0$}{
      $
      \begin{bmatrix}
        x\\
        y
      \end{bmatrix}
      \assign
      \begin{bmatrix}
        x + y \\
        1
      \end{bmatrix}
      $
    }
  }
  The closed form of $x$ is $q = \charfun{n = 0} \cdot x + \charfun{n \neq 0} \cdot (x + y
  + n - 1)$. Replacing $x$ with\linebreak $q_{norm}$ in $x > 0$ yields $x + y + n - 1 > 0$.
  The maximal marked coefficient of $x + y + n - 1$ is $1^{(1,1)}$.
  So by  \Cref{alg:deciding},
    \eqref{eq:nt-loop} does not terminate if $\exists x,y \in \ZZ.\ 1 > 0$ is valid.
  While $1 > 0$ is a tautology, \eqref{eq:nt-loop} terminates if $x \leq 0$ or $x \leq -y$.
   \end{ex}

However, the final formula constructed
by \Cref{alg:deciding} precisely describes
all witnesses for eventual non-termination\report{ (see \Cref{lem:witnesses_proof} for the proof)}.

\begin{lemma}[Witnessing Eventual Non-Termination]
  \label{lem:witnesses}
  Let \eqref{loop} be a triangular loop, let $\vect{q}_{norm}$ be the normalized closed form of \eqref{chained},
  and let
  \[
    \textstyle
    \left(\phi \land \phi[\vect{x} / A\,\vect{x} + \vect{a}]\right)[\vect{x}/\vect{q}_{norm}] = \bigwedge_{i=1}^k p_i > 0.
  \]
  Then $\vect{c} \in \ZZ^d$ witnesses eventual non-termination of \eqref{loop}
  iff $[\vect{x}/\vect{c}]$ is a model for
   \[
    \textstyle
    \bigwedge_{i=1}^k \lia(p_i).
  \]
\end{lemma}
\report{\makeproof{lem:witnesses}{
  We have:
  \[
    \begin{array}{ll}
      &\bigwedge_{i=1}^k \lia(p_i) \,  [\vect{x}/\vect{c}]\\
      \iff&\bigwedge_{i=1}^k \lia(p_i[\vect{x}/\vect{c}])\\
      \iff&\bigwedge_{i=1}^k \base(\max_{\succ}(\coeffs(p_i[\vect{x}/\vect{c}]))) > 0\\
      \iff&\bigwedge_{i=1}^k \lim_{n \mapsto \infty} p_i[\vect{x}/\vect{c}] > 0 \hfill \text{(by  \Cref{lem:asym})}\\
      \iff&\bigwedge_{i=1}^k \exists  n_0 \in \NN.\ \forall n \in
      \NN_{>n_0}.\ p_i[\vect{x}/\vect{c}] > 0 \hfill
      \text{(by  \Cref{thm:limits})}\\
 \iff&\exists  n_0 \in \NN.\ \forall n \in
 \NN_{>n_0}.\
       \left(\phi \land \phi[\vect{x} / A\,\vect{x} + \vect{a}]\right)[\vect{x}/\vect{q}_{norm}]  [\vect{x}/\vect{c}]\\
      \iff&\exists  n_0 \in \NN.\ \forall n \in  \NN_{>n_0}.\ \phi[\vect{x}/\vect{q}_{norm}]  [\vect{x}/\vect{c}] \land \phi[\vect{x} / A\,\vect{x} + \vect{a}][\vect{x}/\vect{q}_{norm}]  [\vect{x}/\vect{c}]\\
      \iff&\exists  n_0 \in \NN.\ \forall n \in  \NN_{>n_0}.\ \phi[\vect{x}/\vect{q}_{norm}]  [\vect{x}/\vect{c}] \land \phi[\vect{x}/\vect{q}_{norm}[n/n+1]]  [\vect{x}/\vect{c}]
          \end{array}
    \] \pagebreak
     \[
     \begin{array}{ll}
         & \hfill (\text{as } \vect{q}_{norm}[n/n+1] = \vect{q}[n/n+1] = f(\vect{q}) = A\,\vect{q} + \vect{a} = A\,\vect{q}_{norm} + \vect{a}\\
      & \hfill \text{for large enough } n)\\
      \iff&\exists  n_0 \in \NN.\ \forall n \in  \NN_{>n_0}.\ \phi[\vect{x}/\vect{q}_{norm}]  [\vect{x}/\vect{c}]\\
      \iff& \vect{c} \text{ witnesses eventual non-termination of \eqref{loop}} \hfill \text{(by  \Cref{cor:normalized})}
    \end{array}
  \]
}}

\section{Conclusion}

We presented a decision procedure for termination of affine integer loops with
triangular update matrices. In this way, we contribute to the ongoing challenge
of proving the 15 years old conjecture by Tiwari \cite{DBLP:conf/cav/Tiwari04} that termination of affine integer loops is
decidable.
After linear
loops \cite{DBLP:conf/cav/Braverman06}, loops with at most $4$
variables \cite{DBLP:conf/soda/OuakninePW15}, and loops with diagonalizable
update matrices \cite{Bozga14,DBLP:conf/soda/OuakninePW15}, triangular loops are the
fourth important special case where decidability could be proven.

The key idea of our decision procedure is to compute \emph{closed forms} for the
values of the program variables after a symbolic number of iterations $n$. While
these closed forms are rather complex, it turns out that reasoning about first-order formulas over the theory of linear integer arithmetic suffices  to analyze their behavior for large $n$. This allows us to
reduce (non-)termination of triangular loops to integer programming.
  In future work, we plan to investigate generalizations of our approach
to other classes of integer loops.

\bibliographystyle{splncs03}
{\normalsize \bibliography{ms}}

 \report{
 \newpage

 \section*{\huge Appendix}

 \appendix

 \section{Example \ref{ex:sum} -- Missing Steps}

\label{ex:sum-full}

\begin{ex}
  We complete \Cref{ex:sum} by showing how to compute poly-exponential expressions $q_1$ and $q_3$ that are equivalent to $p_1 = \sum_{i=1}^{n}
  \charfun{i-1=0} \cdot 4^{n-i} \cdot 2 \cdot w$ and $p_3 =
  \sum_{i=1}^{n} 4^{n-i} \cdot (- 2)$, respectively.
  For $p_1$, according to \eqref{eq:positive-literal} we have
  \[
    \begin{array}{lll}
      &&       \sum_{i=1}^{n}
         \charfun{i-1=0} \cdot 4^{n-i} \cdot 2 \cdot w\\
      &=& \sum_{i=1}^{n} \charfun{n > i-1} \cdot
          \charfun{i-1=0} \cdot 4^{n-i} \cdot 2 \cdot w\\
      &=& \charfun{n \neq 0} \cdot 4^{n-1} \cdot
          2 \cdot w\\
      &=& \charfun{n \neq 0} \cdot \tfrac{1}{2} \cdot w \cdot 4^{n}\\
      &=& q_1.
    \end{array}\]
  For $p_3$, according to \eqref{eq:negative-3} we obtain
  \[
    \begin{array}{lll}
      &&      \sum_{i=1}^{n} 4^{n-i} \cdot (- 2)\\
      &=&  (- 2) \cdot 4^n \cdot \sum_{i=1}^{n} \left(\tfrac{1}{4}\right)^i\\
      &=&  (- 2) \cdot 4^n \cdot \sum_{i=1}^{n} \left(-\tfrac{1}{3} - 4 \cdot
          \left(-\tfrac{1}{3}\right)\right) \cdot \left(\tfrac{1}{4}\right)^i\quad \text{(thus, $r =
          -\tfrac{1}{3}$, cf.\ \Cref{ex:PolynomialDifference})}\\
      &=&  (- 2) \cdot 4^n \cdot \left(
          \sum_{i=1}^{n} \left(-\tfrac{1}{3}\right)  \cdot \left(\tfrac{1}{4}\right)^i \; - \;
          \sum_{i=1}^{n} 4 \cdot \left(-\tfrac{1}{3}\right)  \cdot \left(\tfrac{1}{4}\right)^i
          \right)\\
      &=&  (- 2) \cdot 4^n \cdot \left(
          \sum_{i=1}^{n} \left(-\tfrac{1}{3}\right)  \cdot \left(\tfrac{1}{4}\right)^i \; - \;
          \sum_{i=0}^{n-1}  \left(-\tfrac{1}{3}\right)  \cdot \left(\tfrac{1}{4}\right)^i
          \right)\\
      &=&  (- 2) \cdot 4^n \cdot
          \charfun{n \neq 0} \cdot
          \left(\left(-\tfrac{1}{3}\right)  \cdot \left(\tfrac{1}{4}\right)^n \; - \;
          \left(-\tfrac{1}{3}\right)
          \right)\\
      &=&
          \charfun{n \neq 0} \cdot
          \left(\tfrac{2}{3}  -
          \tfrac{2}{3} \cdot 4^n\right)\\
      &=&  \tfrac{2}{3}  -
          \tfrac{2}{3} \cdot 4^n\\
      &=& q_3.
    \end{array}\]
\end{ex}

\section{From Non-Termination to Eventual Non-Termination}
\label{app:non-term}

\begin{lemma}
  \eqref{loop} is non-terminating $\iff$ \eqref{loop} is eventually non-ter\-mi\-nat\-ing.
\end{lemma}
\begin{proof}
  The ``only if'' direction is trivial by choosing $n_0 = 0$ in \Cref{def:non-term}.
  For the ``if'' direction, assume that \eqref{loop} is eventually non-terminating, i.e., there are $\vect{c} \in \ZZ^{d}, n_0 \in \NN$ such that
  \[
    \forall n \in \NN_{>n_0}.\ \phi[\vect{x} / f^{n}(\vect{c})],
  \]
  cf.\ \Cref{def:non-term}.
  Let $\vect{c}' = f^{n_0 + 1}(\vect{c})$.
  Then
  \[
    \forall n \in \NN.\ \phi[\vect{x} / f^n(\vect{c}')],
  \]
  i.e., \eqref{loop} is non-terminating, see \Cref{def:term}.
\end{proof}

 }
\end{document}